\documentclass{ifacconf}

\usepackage{graphicx,amsmath}      
\usepackage{natbib}        
\usepackage{booktabs}
\usepackage{siunitx}
\usepackage{graphicx}
\usepackage{amsmath}
\usepackage{amsfonts}
\usepackage{mathrsfs}
\newcommand{\End}{\hfill $\square$}

\usepackage{booktabs}
\usepackage{multirow}
\usepackage{adjustbox} 

\newtheorem{theorem}{Theorem}
\newtheorem{corollary}{Corollary}
\newtheorem{lemma}{Lemma}
\newtheorem{proposition}{Proposition}

\newtheorem{remark}{Remark}

\newtheorem{assumption}{Assumption}

\DeclareFontFamily{U}{mathx}{\hyphenchar\font45}
\DeclareFontShape{U}{mathx}{m}{n}{
	<5> <6> <7> <8> <9> <10>
	<10.95> <12> <14.4> <17.28> <20.74> <24.88>
	mathx10
}{}
\DeclareSymbolFont{mathx}{U}{mathx}{m}{n}
\DeclareMathSymbol{\bigtimes}{1}{mathx}{"91}

\NewDocumentCommand{\tra}{om}{%
	\IfNoValueTF{#1}
	{#2}
	{#2_{[#1]}}%
}

\usepackage{tikz}
\usetikzlibrary{shadows,arrows.meta,positioning,backgrounds,fit,chains,scopes,calc,shapes}
\usetikzlibrary{backgrounds}
\definecolor{gray1}{rgb}{0.9,	0.91,	0.93}
\definecolor{lightblue}{rgb}{0.6, 0.76, 0.92}
\definecolor{gray2}{rgb}{0.95, 0.95,0.95	}
\definecolor{gray3}{rgb}{0.85, 0.85, 0.85	}
\definecolor{lightblue1}{rgb}{0.87, 0.92, 0.97	}
\definecolor{lightblue2}{rgb}{0.62, 0.76, 0.90	}
\definecolor{lightblue3}{rgb}{0.36	0.61	0.84}

\newcommand{\R}{\mathbb{R}}

\newcommand{\Tf}{ T}
\newcommand{\Tp}{ T_\text p}
\newcommand{\temp}{ \Theta}

\begin{document}
\begin{frontmatter}

\title{Optimal energy conversion -- A port-Hamiltonian approach}

\author[First]{Guanru Pan} 
\author[Second]{Karl Worthmann} 
\author[Third]{Manuel Schaller} 
\author[First]{Timm Faulwasser}

\address[First]{
	 Hamburg University of Technology, Hamburg, Germany (e-mails: guanru.pan@tuhh.de, timm.faulwasser@ieee.org).}
\address[Second]{
TU Ilmenau, Ilmenau, Germany (e-mail: karl.worthmann@tu-ilmenau.de).}
\address[Third]{
TU Chemnitz, Chemnitz, Germany (e-mail: manuel.schaller@math.tu-chemnitz.de).}
\begin{abstract}
This paper studies optimal energy conversion in port-Hamiltonian systems. We analyze state- and control-constrained finite-horizon optimal control problems that maximize energy conversion efficiencies. Under suitable reachability  conditions, we show that the asymptotic optimal efficiencies exist and coincide with the corresponding optimal cyclic efficiencies. For a class of port-Hamiltonian systems describing coupled storage-conversion components, we further derive intrinsic limits on achievable energy conversion and establish a generalized Carnot-type efficiency bound. Our results provide an optimal-control framework connecting asymptotic optimal operation, cyclic operation, and fundamental efficiency limits for port-Hamiltonian systems.
\end{abstract}


\end{frontmatter}

\section{Introduction}
Energy conversion is a fundamental mechanism in a wide range of engineering systems, including electromechanical devices, power electronic converters, and thermodynamic processes. In thermodynamics, cyclic operation provides a natural paradigm for analyzing energy conversion, where operating cycles are typically designed a priori and idealized as quasi-static to characterize the maximum achievable efficiency. Modern engineering systems, however, often operate dynamically under state, control, and physical constraints, raising the question of how such dynamic operation can be optimally controlled 
to achieve maximal
energy-conversion efficiency.
In this context, energy-based modeling provides a natural framework for describing conversion systems and their interactions. In particular, port-Hamiltonian (pH) systems offer a structured representation of energy storage, interconnection, dissipation, and external power exchange, thereby enabling a unified treatment of multi-physics systems~\citep{Duindam_2009, Schaft2014}. The pH framework has proven particularly effective for the design of passivity-based controllers~\citep{Ortega_2002,Zhang2018,Macchelli2023}.

On the control side, optimization-based methods have been widely used to improve the performance for constrained dynamical systems~\citep{Mayne2014}. Recent contributions have extended optimal control formulations to pH systems by incorporating energy-, entropy-, and exergy-related objectives~\citep{Schaller2021,Philipp2024}. 
Moreover, dissipativity-based analysis and turnpike theory have revealed that long-horizon optimal operation is often closely related to steady-state or periodic behavior~\citep{damm2014exponential,Zanon2017,Koehler2018,Tr_lat_2018,Faulwasser2022a}.

Energy conversion in port-Hamiltonian systems has recently been studied
from a structural perspective
\citep{Schaft2021,Schaft2022}.
 These works relate the
transfer of energy between two ports to the interconnection
topology.
In particular, without direct interconnection between the
state components associated with injection and
extraction, energy cannot be converted while keeping the
injection-side output constant. This motivates generalized
Carnot cycles alternating between different constant
injection-side outputs and constant-state phases.

More broadly, \cite{Schaft2021a} motivates extending
concepts from classical thermodynamics to systems and control. There,
First and Second Laws of Thermodynamics are interpreted in terms of cyclo-losslessness
and cyclo-dissipativity, respectively, and the classical Carnot bound is
derived from the Kelvin formulation of the Second Law. This treatment
raises the question of whether analogous bounds can be established for
two-port systems beyond thermodynamics.

Despite these developments, optimal operation  has not been characterized.
Indeed, \cite{Schaft2021} explicitly identify the
formulation of optimal energy conversion and the extension of Carnot
cycles to more general cyclic processes with a notion of efficiency as
open problems. Moreover, although a Carnot-type efficiency is derived
for the generalized Carnot construction, whether it bounds the optimal
efficiency remains open.

Motivated by these observations, this paper investigates constrained energy conversion in two-port port-Hamiltonian systems from an optimal control perspective. We study the relation between finite-horizon optimal operation, optimal cyclic operation, and fundamental efficiency limitations.
The main contributions of this paper are summarized as follows:
i) We formulate finite-horizon optimal control problems for two complementary notions of energy conversion efficiency and characterize their asymptotic optimal performance.
ii) Under suitable reachability conditions, we prove that the asymptotic optimal efficiencies coincide with the corresponding optimal cyclic efficiencies.
iii) Building upon the energy-conversion framework of~\cite{Schaft2021,Schaft2022}, we investigate a class of block-diagonal port-Hamiltonian systems representing coupled storage-conversion components. We identify structural restrictions on achievable energy conversion and, under an additional scalar structural restriction, establish a generalized Carnot-type efficiency as a rigorous upper bound.

The remainder of the paper is organized as follows. Section~\ref{sec:twoport}
states the energy conversion problem within the two-port port-Hamiltonian
framework. Section~\ref{sec:OCP} formulates the corresponding optimal control
problems and derives general bounds on the conversion efficiencies using the
available storage function. Section~\ref{sec:asymptotic} establishes the
asymptotic characterization of the optimal efficiencies and their relation to
optimal cyclic operation. Section~\ref{sec:limits} investigates energy
conversion in block-diagonal port-Hamiltonian systems, while Section~\ref{sec:carnot} establishes a
generalized Carnot-type efficiency bound under an additional scalar structural
restriction. Section~\ref{sec:simulation} presents numerical examples, and
Section~\ref{sec:conclusion} concludes the paper.

\section{Problem statement} \label{sec:twoport}

We consider nonlinear port-Hamiltonian (pH) systems in two-port form, 
\begin{subequations}\label{eq:PH}
\begin{align}
\dot x &= \big(J(x)-R(x)\big)H_x(x)
+ \begin{bmatrix}
    B_1^\top(x) \\ B_2^\top(x)
\end{bmatrix}^\top \begin{bmatrix} u_1 \\ u_2 \end{bmatrix},  
 \\
\begin{bmatrix} y_1 \\ y_2 \end{bmatrix} 
&= \begin{bmatrix}
    B_1^\top(x) \\ B_2^\top(x)
\end{bmatrix}H_x(x), 
\end{align}
\end{subequations}
with initial condition $x(0)=x_0 \in \mathbb R^{n_x}$. Here, the interconnection matrix  $J(x)=-J^\top(x)$, the dissipation matrix  $R(x)=R^\top(x)\succeq 0$, and the input matrix $B(x)\doteq \begin{bmatrix}
    B_1(x) & B_2(x)
\end{bmatrix}\in\mathbb{R}^{n_x\times(m_1+m_2)}$ may depend on the state $x \in\mathbb{R}^{n_x}$. The maps $J:\mathbb R^{n_x} \to\mathbb R^{n_x\times n_x}  $, $R:\mathbb R^{n_x} \to\mathbb R^{n_x\times n_x}$, and $B:\mathbb R^{n_x} \to\mathbb R^{n_x\times (m_1+m_2)}$ are assumed to be locally Lipschitz continuous. The continuously differentiable Hamiltonian $H:\mathbb{R}^{n_x}\to\mathbb{R}_{\geq 0}$ represents the total stored energy of the system, while $H_x \doteq \nabla_x H:\mathbb{R}^{n_x}\to\mathbb{R}^{n_x}$ denotes the gradient of the Hamiltonian with respect to the state~$x$.

As depicted in Fig.~\ref{fig:2port}, the system interacts with its environment through two external ports: an \emph{energy-injection port} with power variables $(u_1,y_1)\in\mathbb{R}^{m_1}\times \mathbb{R}^{m_1}$ and an \emph{energy-extraction port} with $(u_2,y_2)\in \mathbb{R}^{m_2}\times \mathbb{R}^{m_2}$. We study the problem of \emph{dynamic energy conversion} between these two ports. 
\begin{figure}
    \centering
    \includegraphics[height=0.4\linewidth]{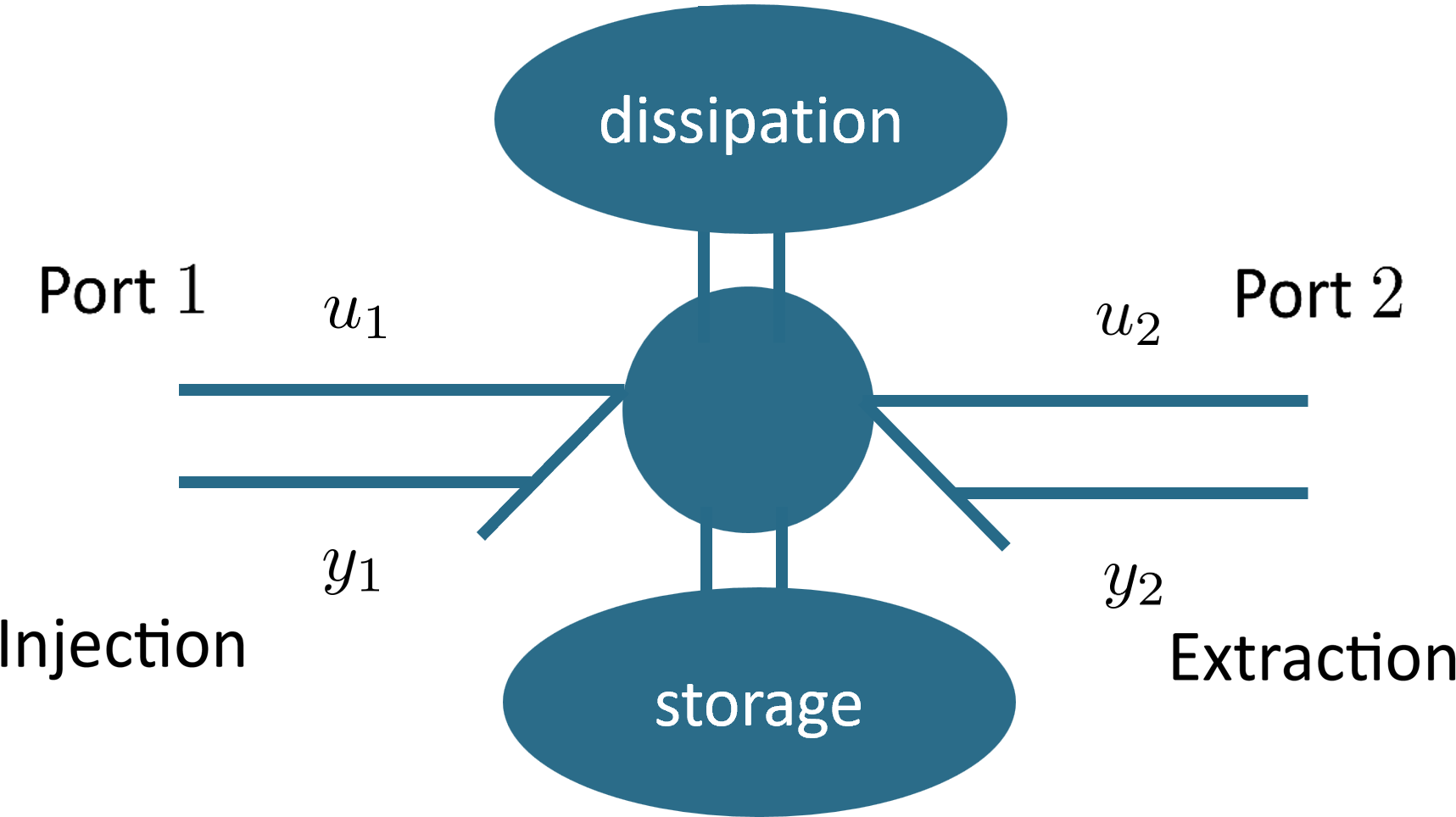}
    \caption{A two-port port-Hamiltonian system with Port~1 for energy injection and Port~2 for energy extraction.}
    \label{fig:2port}
\end{figure}
Directly from the structure of system~\eqref{eq:PH}, trajectories satisfy the \emph{energy balance} and the related \emph{passivity inequality}, for $T \in \mathbb R^+$,
\begin{align}
    &H(x(\Tf)) - H(x(0)) \nonumber \\& = \int_0^{\Tf} - H_x(x(t))^\top R(x(t)) H_x(x(t)) + \left[\begin{smallmatrix}
        u_1(t)\\
        u_2(t)
    \end{smallmatrix}\right]^\top \left[\begin{smallmatrix}
        y_1(t)\\
        y_2(t)
    \end{smallmatrix}\right]\,\mathrm{d}t \nonumber \\
    &\leq \int_0^{\Tf} u_1^\top(t) y_1(t) + u_2^\top(t) y_2(t)\,\mathrm{d}t. \label{eq:supplied} 
\end{align}

We emphasize that the distinction between \emph{injected} and
\emph{extracted} energy is specific to the conversion setting
considered here. Formally, the passivity inequality depends only on the
total supply rate
\[
u^\top y = u_1^\top y_1 + u_2^\top y_2,
\]
and treats both ports symmetrically. The distinction between energy injection and extraction ports is
therefore \emph{semantic rather than intrinsic}: it is identified by the
user according to the desired energy-conversion interpretation and does
not affect the underlying passivity property. However, from a thermodynamical perspective, this semantic difference allows to investigate problems of energy conversion through the lens of open control systems.

For a state-output trajectory \((x(\cdot),y(\cdot))\) resulting from 
initial
condition \(x(0)=x_0\) and input $u(\cdot)$, we define the energy injected at Port~1 and the energy
extracted at Port~2 as
\begin{subequations}\label{eq:port_energy}
\begin{align}
E_1(x_0,u(\cdot))
&\doteq
\int_0^{\Tf} u_1^\top(t)y_1(t)\,\mathrm dt,\\
E_2(x_0,u(\cdot))
&\doteq
-\int_0^{\Tf} u_2^\top(t)y_2(t)\,\mathrm dt.
\end{align}
\end{subequations}

Whenever \(E_1(x_0,u(\cdot))>0\), the conversion efficiency is
defined as
\begin{equation}\label{eq:eff_true}
\eta(x_0,u(\cdot))
\doteq
\frac{E_2(x_0,u(\cdot))}
{E_1(x_0,u(\cdot))}.
\end{equation}
Besides $\eta$, we also consider  the
\emph{generalized thermal efficiency}
\begin{equation}\label{eq:gamma}
{\eta_{+}}(x_0,u(\cdot))
\doteq
\frac{E_2(x_0,u(\cdot))}
     {E_{1,+}(x_0,u(\cdot))},    
\end{equation}
where $[a]_+ := \max\{a,0\}$ is the positive part of $a\in \R$ and
\begin{equation}\label{eq:E1plus}
E_{1,+}(x_0,u(\cdot))
\doteq
\int_0^{\Tf}
\left[u_1^\top(t)y_1(t)\right]_+\,\mathrm{d}t \geq 0
\end{equation}
denotes the energy accumulated from the positive power transfer through
Port~1. As $E_{1,+} \geq E_1$, we have
\begin{equation}~\label{eq:efficency_comparison}
0\leq{\eta_{+}} \leq \eta \quad \text{if} \quad E_1> 0 \text{ and } E_2 \geq 0.
\end{equation}

The generalized thermal efficiency efficiency $\eta_+$ distinguishes energy supplied
through Port~1 from energy returned through the same port. To motivate
this terminology, consider a heat engine for which
$u_1^\top y_1=\dot Q$ is the heat-transfer rate and
$-u_2^\top y_2=\dot W$ is the mechanical power output. Then
\[
E_{1,+}=Q_{\mathrm{in}}
\doteq\int_0^T[\dot Q(t)]_+\,\mathrm{d}t,
\,\,
E_2=W_{\mathrm{net}}
\doteq\int_0^T \dot W(t)\,\mathrm{d}t,
\]
and hence
$ \displaystyle
\eta_+
=\frac{W_{\mathrm{net}}}{Q_{\mathrm{in}}}
$
corresponds to the conventional thermal efficiency of a heat engine; see \citep{Struchtrup_2024}. Here,
the heat rejected by the system,
$ \displaystyle
Q_{\mathrm{out}}
\doteq\int_0^T[-\dot Q(t)]_+\,\mathrm{d}t,
$
is regarded as heat dissipated to
the environment, i.e. waste heat. For general two-port systems, $\eta_+$ extends this
input--output distinction beyond the thermodynamic setting.

In this paper, we investigate the optimal energy conversion problem by maximizing the conversion efficiencies $\eta$ and $\eta_{+}$ subject to input and state constraints.
Let $\mathbb X \subset \mathbb R^{n_x}$ and
$\mathbb U \subset \mathbb R^{m_1+m_2}$ be compact sets containing the
origin in their interiors. Further, let $\mathbb U$ be convex.\footnote{Formally, one could consider
general mixed state-input constraints $(x,u)\in\mathbb Z$ with compact
$\mathbb Z$.}
For a finite horizon $T\in\mathbb R^+$, we consider a measurable input function
$u\in L^1(0,T;\mathbb R^{m_1+m_2})$ and corresponding absolutely continuous state function
$x\in W^{1,1}(0,T;\mathbb R^{n_x})$ satisfying \footnote{We have $u \in L^1(0,T;\mathbb R^{m_1 + m_2})$ and $x\in W^{1,1}(0,T;\mathbb R^{n_x})$, i.e., the input as well as the state trajectory and its derivative are absolutely integrable.}  
\[
u(t)\in\mathbb U,\qquad x(t)\in\mathbb X
\qquad \text{for all }t\in[0,T].
\]

Due to compactness of $\mathbb X$ and continuity of $H$, the Hamiltonian
is uniformly bounded, i.e., there exists $\overline H\in\mathbb R_+$
such that
\begin{equation}\label{eq:barH}
    0\leq H(x)\leq\overline H
    \qquad \forall\,x\in\mathbb X.
\end{equation}

Moreover, since \(B\) and \(H_x\) are continuous and \(\mathbb X\times\mathbb U\)
is compact, the exchanged power at both ports is uniformly bounded. In particular,
\begin{subequations}\label{eq:p_bound}
\begin{align}
 p_1
&\doteq
\max_{(x,u)\in\mathbb X\times\mathbb U}
|u_1^\top B_1(x)H_x(x)|,\\
 p_2
&\doteq
\max_{(x,u)\in\mathbb X\times\mathbb U}
|u_2^\top B_2(x)H_x(x)|
\end{align}
\end{subequations}
are finite. By definition, for all
\(u\in\mathbb U\) and \(x\in\mathbb X\),
\begin{align*}
    |u_1^\top y_1|= |u_1^\top B_1(x)H_x(x)|\le  p_1,\\
|u_2^\top y_2| =|u_2^\top B_2(x)H_x(x)|\le  p_2.
\end{align*}

\section{Optimal energy conversion problems}\label{sec:OCP}

For $\phi \in \{\eta, {\eta_{+}}\}$, we consider the following Optimal Control Problem (OCP) that aims to maximize the corresponding conversion efficiency for a given horizon $T \in \mathbb R^+$, i.e., we either maximize the energy conversion rate~\eqref{eq:eff_true} or the generalized thermal efficiency~\eqref{eq:gamma} 
\begin{subequations}\label{eq:OCP_Efficiency}
    \begin{align}
        & \phi^\star(x_0,\Tf) \doteq
        \max_{u\in L^1(0,T;\mathbb{R}^{m_1+m_2})}
        \phi(x_0, u(\cdot))
        \\
        \text{s.t.  }\quad&
        \text{system dynamics~\eqref{eq:PH},} \quad x(0)=x_0,\\
        &
        u(t)\in\mathbb U,\;
        x(t)\in\mathbb X
        \quad \forall\,t\in[0,\Tf],\label{eq:state_input_constraints}\\
        &
        -\int_0^{\Tf}u_2^\top(t)y_2(t)\,\mathrm{d}t
        \ge p\Tf. \label{eq:extraction_requirement}
    \end{align}
\end{subequations}
For $\textrm{OCP}$~\eqref{eq:OCP_Efficiency}, the existence of an optimal control is shown in~Proposition~\ref{prop:existence_optimal_controls}, see appendix. 
With $p \in \mathbb R^+$, the minimum energy extraction constraint in \eqref{eq:extraction_requirement} guarantees that the optimal conversion efficiency $\phi^\star$ is positive for sufficiently large horizons $\Tf$, while simultaneously ensuring that the optimization problem is well-posed; see Proposition~\ref{pro:nE_Sa}. Such a minimum extraction requirement is common in energy conversion applications, where a prescribed amount of extracted energy must be achieved while maximizing the conversion efficiency.

To establish bounds on the optimal conversion efficiencies \(\eta^\star\) and $\eta_{+}^\star$, we first revisit the notion of \emph{available storage} from dissipativity theory \citep{Willems1972}. For the port-Hamiltonian system~\eqref{eq:PH}, the available storage is defined as
\begin{subequations}\label{eq:avaiable_storage}
\begin{align}
S_a(x_0) \doteq\;&
\sup_{\substack{u\in L^1(0,T;\mathbb{R}^{m_1+m_2}), \Tf>0}}
-\int_0^{\Tf}u^\top(t)y(t)\,\mathrm{d}t\\
\text{s.t.}\;&
\text{system dynamics~\eqref{eq:PH}, $x(0)=x_0$,}\\
&
u(t)\in\mathbb U,\;
x(t)\in\mathbb X
\quad
\forall\,t\in[0,\Tf].
\end{align}
\end{subequations}
It characterizes the maximum net energy that can be extracted from the system through  both ports, as a function of
the initial state \(x_0\). By the passivity inequality~\eqref{eq:supplied}, we have 
\[
0\le S_a(x_0)\le H(x_0)-\inf_{x\in\mathbb X}H(x);
\]
see~\citep{Willems1972,Schaft2017}. Hence, by~\eqref{eq:barH}, $S_a$ is bounded on $\mathbb X$ . 

Available storage in OCP~\eqref{eq:avaiable_storage} and OCP~\eqref{eq:OCP_Efficiency} describe different
energy-conversion settings. While $S_a$ quantifies the energy that can
be extracted by exploiting the energy initially stored in the system,
OCP~\eqref{eq:OCP_Efficiency} considers sustained conversion enabled by
energy injection through Port~1. In particular, when \(H(x_0)=0\) and \(R(x)\succ0\), passivity
implies that the optimal value of~\eqref{eq:avaiable_storage} is attained
by \(u_1(\cdot)\equiv u_2(\cdot)\equiv0\), since no net energy can be
extracted from the system. Nevertheless,
OCP~\eqref{eq:OCP_Efficiency} may still admit nontrivial and meaningful
optimal inputs \(u_1(\cdot)\) and \(u_2(\cdot)\), as the external energy
injected through Port~1 can ensure and maintain a sustainable
energy-conversion process.

For a finite, but sufficiently large horizon \(\Tf\in\mathbb R^+\), the available
storage consequently provides  bounds on the achievable efficiencies.

\begin{proposition}
[Efficiency bounds via available storage]
\label{pro:nE_Sa}~\\
Consider $\mathrm{OCP}$~\eqref{eq:OCP_Efficiency}. If
\begin{equation}\label{eq:T_requirement}
\Tf>\frac{S_a(x_0)}{p},
\end{equation}
then
\begin{equation}\label{eq:eta_true_upperbound}
\frac{p}{p_1}
\le
\eta_{+}^\star(x_0,\Tf)
\le
\eta^\star(x_0,\Tf)
\le
\frac{p\Tf}{p\Tf-S_a(x_0)},
\end{equation}
where $p_1$ is defined in \eqref{eq:p_bound}.

Consequently, for all $x_0\in\mathbb X_0$,
\begin{subequations}\label{eq:eta_lim_bounds}
\begin{align}
\frac{p}{p_1}
&\le
\liminf_{\Tf\to\infty}\eta_{+}^\star(x_0,\Tf)
\le
\liminf_{\Tf\to\infty}\eta^\star(x_0,\Tf)
\le 1,
\\
\frac{p}{p_1}
&\le
\limsup_{\Tf\to\infty}\eta_{+}^\star(x_0,\Tf)
\le
\limsup_{\Tf\to\infty}\eta^\star(x_0,\Tf)
\le 1.
\end{align}
\end{subequations}
    
\end{proposition}

\begin{pf}
By the definition of available storage in~\eqref{eq:avaiable_storage},
every feasible trajectory $(u,y)$ of OCP~\eqref{eq:OCP_Efficiency} satisfies
\[
-\int_0^{\Tf}u^\top(t)y(t)\,\mathrm dt
\le
S_a(x_0).
\]
Using the definitions of
$E_1(x_0,u(\cdot))$ and $E_2(x_0,u(\cdot))$, this inequality becomes
\begin{equation}\label{eq:port_decomposition}
-E_1(x_0,u(\cdot))
+
E_2(x_0,u(\cdot))
\le
S_a(x_0).
\end{equation}
Rearranging yields
\[
E_1(x_0,u(\cdot))
\ge
E_2(x_0,u(\cdot))-S_a(x_0)
\ge
p\Tf-S_a(x_0)
>
0,
\]
where the second inequality follows from the minimum energy extraction
constraint~\eqref{eq:extraction_requirement}, and the last inequality
follows from \eqref{eq:T_requirement}.

Moreover, by the uniform power bounds~\eqref{eq:p_bound},
\[
E_{1,+}(x_0,u(\cdot))
=
\int_0^{\Tf}[u_1^\top(t)y_1(t)]_+\,\mathrm dt
\le
p_1\Tf.
\]
Hence, every feasible trajectory of OCP~\eqref{eq:OCP_Efficiency} satisfies
\[
{\eta_{+}}(x_0,u(\cdot))
=
\frac{E_2(x_0,u(\cdot))}
     {E_{1,+}(x_0,u(\cdot))}
\ge
\frac{p\Tf}{p_1\Tf}
=
\frac{p}{p_1}.
\]

Furthermore, 
\[
\eta(x_0,u(\cdot))
=
\frac{E_2(x_0,u(\cdot))}
     {E_1(x_0,u(\cdot))}
\le
\frac{E_2(x_0,u(\cdot))}
     {E_2(x_0,u(\cdot))-S_a(x_0)}.
\]
Since the function
\[
f(E)=\frac{E}{E-S_a(x_0)}
\]
is monotonically decreasing for $E>S_a(x_0)$, and
as $E_2(x_0,u(\cdot))\ge p\Tf$, we obtain
\[
\eta(x_0,u(\cdot))
\le
\frac{p\Tf}{p\Tf-S_a(x_0)}.
\]
Since the lower bound on ${\eta_{+}}$ and the upper bound on $\eta$ hold for
every feasible trajectory, they also hold for the corresponding optimal
trajectories. Together with the efficiency comparison of ${\eta_{+}}$ and $\eta$
in~\eqref{eq:efficency_comparison}, this proves
\eqref{eq:eta_true_upperbound}. Taking the limit inferior and the limit
superior as $\Tf\to\infty$ yields~\eqref{eq:eta_lim_bounds}.
\End
\end{pf}

The upper bound admits a simple physical interpretation.
Over short horizons, part of the extracted energy may originate from the
system's initially stored energy, quantified by $S_a(x_0)$, rather than from the energy injected
through Port~1, allowing the conversion efficiency to
temporarily exceed one. Since the available storage is finite, this
effect vanishes as the horizon increases, and the asymptotic conversion
efficiency is necessarily bounded by one, which  aligns with the
first law of thermodynamics.

\begin{remark}[Environment-driven injection  scenario]
In addition to the manipulable energy injection setting considered above, the input at Port~1 may be determined by the system environment, i.e.,
\[
u_1(t)=w(t), \qquad t\in[0,\infty),
\]
where \(w\) denotes an exogenous input, such as wave-induced forces.
When \(w(t)\) is fixed and known, the optimization problem can be reformulated by imposing the additional constraint \(u_1(t)=w(t)\) in OCP~\eqref{eq:OCP_Efficiency}. In this case, the problem reduces to constrained efficiency maximization under prescribed   exogenous input $u_1$. Consequently, the achievable performance is generally lower than in the fully controlled setting. 
When \(w(\cdot)\) is uncertain with either deterministic bounds or distributional information, the problem~\eqref{eq:OCP_Efficiency} becomes robust or stochastic. 
The asymptotic and efficiency analysis of this environment-driven setting constitutes an interesting direction for future work.
\End
\end{remark}

\begin{remark}[Limited energy supply]
An isoperimetric constraint
\[
    E_{1,+}(x_0,u(\cdot))\leq \bar E_1
    \qquad \text{for all } T>0
\]
instead of~\eqref{eq:extraction_requirement} can impose a fixed energy
budget at the injection side while leaving its power profile as a
decision variable, as relevant, e.g., for finite-horizon batch
operation. However, an upper bound on $E_{1,+}$ alone allows
arbitrarily small energy injection and may lead to degenerate efficiency
maximization, requiring additional conditions such as a prescribed
energy output $E_2\geq\underline E_2$ independent of $T$. Here, we
instead impose the minimum extraction requirement
\eqref{eq:extraction_requirement}, which scales with the horizon, and
investigate its asymptotic relation to periodic operation. Fixed-budget
formulations constitute a different setting, since sustained repetition
of a nontrivial cycle requires the supplied energy to grow with the
number of cycles, and are beyond the scope of this work.
\End
\end{remark}

\section{Asymptotic Optimal Energy Conversion Efficiency}\label{sec:asymptotic}

In the following, we establish sufficient conditions to tighten the asymptotic lower and upper bounds of the optimal conversion efficiencies
$
\eta^\star(x_0,\Tf)
$ and $
\eta_{+}^\star(x_0,\Tf)$ from Prop.~\ref{pro:nE_Sa}
and subsequently discuss conditions under which the limits
\[
\lim_{\Tf\to\infty}\eta^\star(x_0,\Tf),
\qquad
\lim_{\Tf\to\infty}\eta_{+}^\star(x_0,\Tf)
\]
exist. To this end, we focus on recurrentlly operating regimes, namely equilibrium (steady-state) and periodic (cyclic) operation, which serve as  benchmarks for long-horizon energy conversion.

\subsection{Asymptotic lower bound}

For $\phi\in\{\eta,\eta_{+}\}$ and a given period
$\Tp\in\mathbb R^+$, we define the optimal periodic efficiency as
\begin{subequations}\label{eq:OCP_clc_Eff}
\begin{align}
\quad&\phi_\text p(\Tp)\doteq 
\max_{\substack{u\in L^1(0,T;\mathbb{R}^{m_1+m_2}),\,x_0 \in \mathbb X}} \phi(x_0, u(\cdot)) 
  \\[1mm]
\text{s.t.}\quad 
& \text{system dynamics~\eqref{eq:PH}, $x(0)=x(\Tp)=x_0$,} \\
&\qquad \qquad \eqref{eq:state_input_constraints}-\eqref{eq:extraction_requirement} \label{eq:periodic_constraint}
\end{align}
\end{subequations}

In contrast to OCP~\eqref{eq:OCP_Efficiency}, the initial state is a
decision variable subject to the periodicity condition. The existence
of optimal solutions to OCP~\eqref{eq:OCP_clc_Eff} is discussed in the
Appendix. 

We note that steady-state operation can be interpreted as a limiting case of
\eqref{eq:OCP_clc_Eff} corresponding to constant trajectories for any period. In this case the
state and inputs remain constant, and the system satisfies the equilibrium
conditions of the port-Hamiltonian dynamics. The corresponding steady-state
efficiency maximization problem is 
\begin{subequations}\label{eq:steady_state}
\begin{align}
\eta_{+,\mathrm{s}} =\eta_{\mathrm{s}}
\doteq
&\max_{\bar u,\bar y,\bar x}
\frac{-\bar u_2^\top \bar y_2}{\phantom{-}\bar u_1 \bar y_1} \\[1mm]
\text{s.t.}\quad
&(J(\bar x)-R(\bar x))\,H_x(\bar x) + B(\bar x)\bar u = 0, \\
&\bar y = B^\top(\bar x)\,H_x(\bar x), \\
&\bar x \in \mathbb X,\quad \bar u \in \mathbb U, \\
&-\bar u_2^\top \bar y_2 \geq p,
\end{align}
\end{subequations}
where $\bar u \doteq [\,\bar u_1^\top\ \bar u_2^\top\,]^\top$ and
$\bar y \doteq [\,\bar y_1^\top\ \bar y_2^\top\,]^\top$.  
\begin{lemma}[Steady-state feasibility $\Rightarrow$ cyclic feasibility]~\\
\label{lem:ss_implies_clc_feasible}
Assume that the steady-state problem~\eqref{eq:steady_state} is feasible.
Then $\bar u_1^\top\bar y_1>0$, and the periodic
problem~\eqref{eq:OCP_clc_Eff} is feasible. Moreover, its optimal value
satisfies
\[
\frac{p}{p_1}
\le
\eta_{+,\mathrm s}
=
\eta_{\mathrm s}
\le
\eta_{+,\mathrm p}(\Tp)
\le
\eta_{\mathrm p}(\Tp)
\le
1,
\,\,
\forall\,\Tp\in\mathbb R^+.
\tag*{\End}
\]
\end{lemma}
\begin{pf}
Let $(\bar x,\bar u,\bar y)$ be any feasible steady-state triple for
\eqref{eq:steady_state}. By passivity, every steady-state trajectory satisfies
\[
\bar u_1^\top\bar y_1+\bar u_2^\top\bar y_2\ge 0,
\]
which yields
\[
\bar u_1^\top\bar y_1
\ge
-\bar u_2^\top\bar y_2
\ge
p>0.
\]
Hence, $\eta_{+,\text s} = \eta_\text s$ as $\max(\bar u_1^\top\bar y_1,0) = \bar u_1^\top\bar y_1
>0$.

Now fix an arbitrary $\Tp\in\mathbb R^+$ and define the constant trajectories
\[
x(t)\equiv\bar x,\qquad
u(t)\equiv\bar u,\qquad
y(t)\equiv\bar y,
\qquad
\forall\,t\in[0,\Tp].
\]
By the equilibrium constraint in~\eqref{eq:steady_state}, these
trajectories satisfy the port-Hamiltonian dynamics~\eqref{eq:PH}. The
constraints~\eqref{eq:periodic_constraint} hold by feasibility of
\eqref{eq:steady_state}, while the periodicity constraint is trivially
satisfied since $x(0)=x(\Tp)=\bar x$. Therefore,
$(x(\cdot),u(\cdot))$ is feasible for
\eqref{eq:OCP_clc_Eff}, implying that the periodic problem is feasible.

As the constant trajectory is feasible for every
$\Tp>0$, combining~\eqref{eq:eta_true_upperbound}, optimality of the periodic problem yields
\[
\frac{p}{p_1}\leq \eta_{\mathrm s} = \eta_{+,\text s}
\le
\eta_{+,\mathrm p}(\Tp)\le
\eta_{\mathrm p}(\Tp).
\]
Since every feasible periodic trajectory satisfies
\[
H(x(\Tp))=H(x(0)),
\]
the dissipation inequality implies
\[
E_2(x_0,u(\cdot))
\le
E_1(x_0,u(\cdot))
\le
E_{1,+}(x_0,u(\cdot)).
\]
Consequently,
\[
\eta_{+,\mathrm p}(\Tp)
\le
\eta_{\mathrm p}(\Tp)
\le
1.
\]
Combining the above results concludes the proof.
\End
\end{pf}

Suppose that, for a given period $\Tp \in \mathbb R^+$ and for $\phi\in\{\eta,{\eta_{+}}\}$,
OCP~\eqref{eq:OCP_clc_Eff} admits an optimal solution with optimal
cyclic efficiency $\phi_{\mathrm p}(\Tp)$.
Let
$(u_{\mathrm p}^\star(\cdot),x_{\mathrm p}^\star(\cdot))$
denote a corresponding optimal periodic trajectory, and define the
associated optimal orbit by
\[
\mathbb O(\Tp)
\doteq
\left\{
x_{\mathrm p}^\star(t)\in\mathbb X
\,\middle|\,
t\in[0,\Tp]
\right\}.
\]
We assume that at least one optimal cyclic orbit can be reached uniformly from
the initial states of the finite-horizon problem.

\begin{assumption}[Reachability of the optimal cyclic orbit]
\label{ass:reachability}
Let $\mathbb X_0\subset\mathbb X$ be a compact set of initial states.
There exists a constant $\bar T>0$ such that, for every
$x_0\in\mathbb X_0$ and $\phi\in\{\eta,\eta_+\}$, there exist
$t_1\in[0,\bar T]$ and an input trajectory $u(\cdot)$, with
$u(t)\in\mathbb U$ for all $t\in[0,t_1]$, such that the corresponding
state trajectory satisfies
\[
    x(0)=x_0,\,\,
    x(t)\in\mathbb X\,\, \forall\,t\in[0,t_1],\,\,
    x(t_1)\in\mathbb O(\Tp). \tag*{\End}
\]
\end{assumption}

This assumption reflects the practical capability of energy conversion systems to enter a desired  operating regime within finite time via start-up  maneuvers. It excludes pathological cases in which the state or energy constraints render the optimal cyclic operation not reachable. 
\begin{theorem}[Asymptotic lower bound]~\label{thm:asymptotic_lower}
Let Assumption~\ref{ass:reachability} hold. 
Suppose that the corresponding optimal cyclic solution
$(u_{\mathrm p}^\star(\cdot),x_{\mathrm p}^\star (\cdot))$ satisfies
\eqref{eq:periodic_constraint} strictly, i.e., there exists
$\delta>0$ such that
\[
E_{2,\mathrm p}^\star
\doteq
-\int_0^{\Tp}
u_{2,\mathrm p}^{\star\top}(t)
y_{2,\mathrm p}^\star(t)\,\mathrm dt
=
p\Tp+\delta.
\]
Then
\[
\liminf_{\Tf\to\infty}
\phi^\star(x_0,\Tf)
\ge
\phi_{\mathrm p}(\Tp),\quad \text{ for }\phi\in\{\eta,{\eta_{+}}\}.
\]
\end{theorem}
\begin{pf}
By Assumption~\ref{ass:reachability}, there exists an  input
$u^a(\cdot)$ on $[0,t_1]$, with $t_1\le\bar T$, steering the state
from $x_0$ to $x_{\mathrm p}^\star(0)\in\mathbb O(\Tp)$.
For any $\Tf\ge t_1+\Tp$, define
\[
\ell\doteq
\left\lfloor
\frac{\Tf-t_1}{\Tp}
\right\rfloor,
\qquad
t_r\doteq
\Tf-t_1-\ell\Tp\in[0,\Tp),
\]
where $\lfloor\cdot\rfloor$ denotes the floor function.

Construct a feasible input $\tilde u(\cdot)$ by concatenating
(a) the transition input $u^a(\cdot)$,
 (b) $\ell$ repetitions of the optimal cyclic input $u_{\mathrm p}^\star(\cdot)$,
and (c) one final partial cycle of duration $t_r$.

Let $E_a$, $E_b$, and $E_c$ denote the contributions to the extracted energy $E_2$
during the three phases.
By~\eqref{eq:p_bound},
\[
E_a\ge-p_2\bar T,
\qquad
E_c\ge-p_2\Tp,
\]
and therefore
\[
E_2(x_0,\tilde u)
=
E_a+E_b+E_c
\ge
E_b-(\Tp+\bar T)p_2.
\]
Since the periodic phase consists of $\ell$ complete cycles,
\[
E_b
=
\ell E_{2,\mathrm p}^\star
=
\ell(p\Tp+\delta),
\]
which yields
\[
E_2(x_0,\tilde u)
\ge
\ell(p\Tp+\delta)
-(\Tp+\bar T)p_2.
\]
Moreover, using $\Tf=t_1+\ell\Tp+t_r$, we obtain
\[
E_2(x_0,\tilde u)-p\Tf
\ge
\ell\delta
-p(t_1+t_r)
-(\Tp+\bar T)p_2.
\]
Hence, if
\[
\ell
\ge
\frac{p(t_1+t_r)+(\Tp+\bar T)p_2}{\delta},
\]
then
\[
E_2(x_0,\tilde u)\ge p\Tf,
\]
i.e., the extraction constraint~\eqref{eq:extraction_requirement} is satisfied for sufficient large horizons.

For the injected energy, we analogously obtain
\begin{align*}
    E_1(x_0,\tilde u)
\le
\frac{1}{\eta_{\mathrm p}(\Tp)}
\ell E_{2,\mathrm p}^\star
+
(\Tp+\bar T)p_1, \text{ for } \phi = \eta,\\
    E_{1,+}(x_0,\tilde u)
\le
\frac{1}{\eta_{+,\mathrm p}(\Tp)}
\ell E_{2,\mathrm p}^\star
+
(\Tp+\bar T)p_1, \text{ for } \phi = {\eta_{+}}
\end{align*}
where the first term is the injected energy during the $\ell$ complete periodic cycles, and the second term bounds the contribution of the initial transition and the final partial cycle using~\eqref{eq:p_bound}.

Since $\phi^\star(x_0,\Tf)$ is optimal,
\[
\phi^\star(x_0,\Tf) \geq 
\frac{
\ell E_{2,\mathrm p}^\star
-(\Tp+\bar T)p_2
}{
\dfrac{1}{\phi_{\mathrm p}(\Tp)}
\ell E_{2,\mathrm p}^\star
+
(\Tp+\bar T)p_1
}.
\]
Finally, since
$\Tf=t_1+\ell\Tp+t_r$ with
$t_1\le\bar T$ and $t_r<\Tp$, we have
$\ell\to\infty$ as $\Tf\to\infty$, and therefore
\[
\liminf_{\Tf\to\infty}
\phi^\star(x_0,\Tf)
\ge
\phi_{\mathrm p}(\Tp). \tag*{\End}
\]
\end{pf}
The result shows that periodic operation provides an asymptotic lower bound on achievable performance, as the contribution of transient phases becomes negligible for large horizons.

\subsection{Asymptotic upper bound}

We next derive an upper bound on the asymptotic performance of 
$\mathrm{OCP}$~\eqref{eq:OCP_Efficiency}. To this end, for $\phi \in \{\eta, {\eta_{+}}\}$, we introduce the optimal cyclic
efficiency 
\begin{align}
\phi^{\star}_\text p
\doteq \sup_{\substack{\Tp>0}} \phi_\text p(\Tp). \label{eq:OCP_clc_Eff_sup}
\end{align}

The comparison between the finite-horizon
OCP~\eqref{eq:OCP_Efficiency} and the cyclic benchmark
OCP~\eqref{eq:OCP_clc_Eff_sup} requires that optimal finite-horizon
trajectories can be extended to feasible cyclic trajectories.
To this end, we impose the following reachability condition from
optimal terminal states to the initial state.

\begin{assumption}[Reachability of the initial state]
\label{ass:returnability_opt_terminal}
For every $x_0\in\mathbb X_0$ and $\Tf>0$, let
$x_\mathrm f^\star\in\mathbb X$ denote the terminal state of an optimal
solution of OCP~\eqref{eq:OCP_Efficiency}. There exists a constant
$\bar T>0$ such that, for some $\tau_\mathrm f\in[0,\bar T]$, there exists an
input trajectory $u_\mathrm r \in L^1(0,\tau_\mathrm f;\mathbb U)$, whose corresponding state trajectory $x_\mathrm r(\cdot)$
satisfies
\begin{subequations}\label{eq:return}
\begin{align}
&x_\mathrm r(0)=x_\mathrm f^\star,\,\,
x_\mathrm r(\tau_ \mathrm f)=x_0,\,\,
x_\mathrm r(t)\in\mathbb X , \forall\,t\in[0,\tau_\mathrm f],\\
&-\int_0^{\Tf}u_2^{\star\top}(t)y_2^\star(t)\,\mathrm{d}t
-\int_0^{\tau_f}u_{\mathrm r,2}^\top(t)y_{\mathrm r,2}(t)\,\mathrm{d}t
\geq (\Tf+\tau_\mathrm f)p .
\end{align}
\end{subequations}
\End
\end{assumption}

Assumption~\ref{ass:returnability_opt_terminal} requires that the
initial state can be reached from every optimal terminal state within
uniformly bounded time while preserving the extraction constraint.
Hence, the optimal finite-horizon trajectory can be extended to a
feasible cyclic trajectory. The assumption can alternatively be
enforced through suitable terminal constraints in
OCP~\eqref{eq:OCP_Efficiency}, e.g., by restricting the terminal state
to a set from which $x_0$ is reachable.

Consequently, the additional cost associated with steering the terminal state back to the initial condition remains uniformly bounded and therefore becomes negligible for sufficiently long horizons. This permits a comparison between finite-horizon optimal trajectories and the cyclic benchmark~\eqref{eq:OCP_clc_Eff_sup}, and enables us to derive an asymptotic upper bound complementary to Theorem~\ref{thm:asymptotic_lower}.

\begin{theorem}[Asymptotic upper bound]
\label{thm:asymptotic_upper}
Let Assumption~\ref{ass:returnability_opt_terminal} hold. Then, for every
$x_0\in\mathbb X_0$,
\[
\limsup_{\Tf\to\infty}\phi^\star(x_0,\Tf)
\le
\phi_{\mathrm p}^\star,
\qquad
\text{for }\phi\in\{\eta,{\eta_{+}}\}.
\tag*{\End}
\]
\end{theorem}

\begin{pf}
Fix $\phi\in\{\eta,{\eta_{+}}\}$, and let
$(x^\star(\cdot),u^\star(\cdot))$ be an optimal solution of 
OCP~\eqref{eq:OCP_Efficiency} on $[0,\Tf]$, with terminal state
$x_\mathrm{f}^\star=x^\star(\Tf)$. By
Assumption~\ref{ass:returnability_opt_terminal}, there exist
$\tau_\mathrm{f}\in[0,\bar T]$ and an  return trajectory
$(x_\mathrm r(\cdot),u_\mathrm r(\cdot))$ satisfying~\eqref{eq:return}.

Concatenate the optimal trajectory and the return trajectory to obtain
$(x_ \mathrm e(\cdot),u_\mathrm e(\cdot))$ on $[0,\Tf+\tau_\mathrm{f}]$. Since
\[
x_\mathrm e(0)=x_0=x_\mathrm e(\Tf+\tau_\mathrm{f}),
\]
the extended trajectory is cyclic. Moreover, by Assumption~\ref{ass:returnability_opt_terminal},
the extended trajectory is feasible for
OCP~\eqref{eq:OCP_clc_Eff_sup}.
By optimality of $\phi_{\mathrm p}^\star$ for 
OCP~\eqref{eq:OCP_clc_Eff_sup},
\begin{equation}~\label{eq:comparison_phi_p}
\phi_\mathrm e(x_0,\Tf)\le\phi_{\mathrm p}^\star.
\end{equation}

Let
\[
\Delta_2
\doteq
-\int_0^{\tau_\mathrm  f}
u_{\mathrm  r,2}^\top(t)y_{\mathrm  r,2}(t)\,\mathrm dt
\]
and define
\[
D_\phi^\star(x_0,\Tf)
\doteq
\begin{cases}
E_1^\star(x_0,\Tf), & \phi=\eta,\\[1mm]
E_{1,+}^\star(x_0,\Tf), & \phi={\eta_{+}},
\end{cases}
\]
together with
\[
\Delta_\phi
\doteq
\begin{cases}
\displaystyle
\int_0^{\tau_\mathrm{f}}
u_{\mathrm r,1}^\top(t)y_{\mathrm r,1}(t)\,\mathrm dt,
& \phi=\eta,\\[3mm]
\displaystyle
\int_0^{\tau_\mathrm{f}}
\bigl[u_{\mathrm  r,1}^\top(t)y_{\mathrm  r,1}(t)\bigr]_+\,\mathrm dt,
& \phi={\eta_{+}}.
\end{cases}
\]
By~\eqref{eq:p_bound} and $\tau_\mathrm f\le\bar T$, the quantities
$\Delta_2$ and $\Delta_\phi$ are uniformly bounded. Hence, the
efficiency of the extended cyclic trajectory is
\[
\phi_\mathrm e(x_0,\Tf)
=
\frac{E_2^\star(x_0,\Tf)+\Delta_2}
     {D_\phi^\star(x_0,\Tf)+\Delta_\phi}.
\]

By Proposition~\ref{pro:nE_Sa},
\[
\lim_{T\to \infty}E_1^\star(x_0,\Tf)
\ge \lim_{t\to \infty}
p\Tf-S_a(x_0)=\infty.
\]
Since $E_{1,+}^\star\ge E_1^\star$, it follows that
\[
\lim_{T\to \infty}D_\phi^\star(x_0,\Tf)=\infty
\qquad
\text{for all }\phi\in\{\eta,{\eta_{+}}\}.
\]
Because $\Delta_2$ and $\Delta_\phi$ are bounded,
\begin{multline*}
\limsup_{\Tf\to\infty}\phi_\mathrm e(x_0,\Tf)- \limsup_{\Tf\to\infty}\phi^\star(x_0,\Tf)\\
=
\limsup_{\Tf\to\infty}
\frac{E_2^\star(x_0,\Tf)} {D_\phi^\star(x_0,\Tf)} -\frac{E_2^\star(x_0,\Tf)+\Delta_2} 
    {D_\phi^\star(x_0,\Tf)+\Delta_\phi} =0
.
\end{multline*}
By combining the above result with~\eqref{eq:comparison_phi_p}, we have
\[
\limsup_{\Tf\to\infty}\phi^\star(x_0,\Tf)
=\limsup_{\Tf\to\infty}\phi_\mathrm e(x_0,\Tf)\le
\phi_{\mathrm p}^\star
\]
\End
\end{pf}

\begin{assumption}[Existence of an optimal period]\label{ass:Tpstar}
~\\
The supremum defining $\phi_{\mathrm p}^\star$ in
OCP~\eqref{eq:OCP_clc_Eff_sup} is attained for some
period $T_{\mathrm p}^\star>0$. If multiple optimal
periods exist, $T_{\mathrm p}^\star$ denotes the smallest
one. \End
\end{assumption}
By combining the lower bound established in 
Theorem~\ref{thm:asymptotic_lower} and the upper bound in 
Theorem~\ref{thm:asymptotic_upper}, we have the following result. 
\begin{corollary}\label{cor:asymptotic_exact}
Let Assumptions~\ref{ass:reachability}--\ref{ass:Tpstar} hold. Then, for any 
$x_0 \in \mathbb{X}_0$, the asymptotic optimal value of 
$\mathrm{OCP}$~\eqref{eq:OCP_Efficiency} exists and satisfies
\[
\lim_{{\Tf} \to \infty} \phi^\star(x_0,\Tf)
=
\phi^{\star}_\text p. \tag*{\End}
\]
\end{corollary}

The preceding results are not specific to the port-Hamiltonian
representation. In fact, they extend directly to general dissipative
systems whenever the supply rate admits a two-port decomposition, as
summarized in the following corollary.

\begin{corollary}[Extension to two-port dissipative systems]
Consider a dynamical system that is dissipative with respect to the
separable supply rate
\[
    s(u,y)=s_1(u_1,y_1)+s_2(u_2,y_2),
\]
and define
\begin{align*}
    E_1(x_0,u(\cdot))
    &\doteq \int_0^T s_1(u_1(t),y_1(t))\,\mathrm{d}t,\\
    E_2(x_0,u(\cdot))
    &\doteq -\int_0^T s_2(u_2(t),y_2(t))\,\mathrm{d}t,\\
    E_{1,+}(x_0,u(\cdot))
    &\doteq \int_0^T [s_1(u_1(t),y_1(t))]_+\,\mathrm{d}t.
\end{align*}
Suppose that $s_1$ and $s_2$ are uniformly bounded along all feasible
trajectories and that the corresponding assumptions of
Prop.~\ref{pro:nE_Sa}, Lem.~\ref{lem:ss_implies_clc_feasible},
Thm.~\ref{thm:asymptotic_lower}--\ref{thm:asymptotic_upper}, and
Cor.~\ref{cor:asymptotic_exact} hold. Then the conclusions of these
results remain valid with the above definitions of
$E_1$, $E_2$, and $E_{1,+}$.
\end{corollary}

While the results above thus hold for general two-port dissipative
systems, the port-Hamiltonian representation provides additional
structural insight into the achievable asymptotic conversion
efficiency. In particular, if $R=0$, i.e., the system is lossless,
every feasible cyclic trajectory satisfies $\eta=1$. Hence, the
existence of a feasible cyclic operation implies
\[
    \eta_{\mathrm p}^\star=1.
\]
Consequently, by Cor.~\ref{cor:asymptotic_exact},
\[
    \lim_{\Tf\to\infty}\eta^\star(x_0,\Tf)=1.
\]

The efficiency $\eta$ therefore does not reveal fundamental
limitations of lossless energy conversion. In contrast, $\eta_+$
normalizes the extracted energy only by the positive injected energy.
Thus, whenever $u_1^\top(t)y_1(t)<0$ over parts of a cycle,
$\eta_+$ may admit nontrivial upper bounds strictly below one even in
the lossless case. The next section exploits the port-Hamiltonian
structure to derive conditions under which such bounds admit a
generalized Carnot-type interpretation.

\section{Intrinsic conversion limits in block-diagonal form}\label{sec:limits}

We now investigate intrinsic limitations on the achievable energy
conversion imposed by the port-Hamiltonian system structure. To this
end, we consider the following class of two-port port-Hamiltonian
systems with block-diagonal interconnection structure:
\begin{subequations}\label{eq:PH_block}
\begin{align}
\dot x_1 &=  \left(J_1(x)-R_1(x)\right)H_{x_1}(x) + G_1(x) u_1,\\
\dot x_2& =  \left(J_2(x)-R_2(x)\right)H_{x_2}(x) + G_2(x) u_2,\\
y_1 &= G_1^\top(x) H_{x_1}(x),\\
y_2 &= G_2^\top(x) H_{x_2}(x),
\end{align}
\end{subequations}
where
\(H_{x_1}(x)\doteq \nabla_{x_1}H(x)\) and
\(H_{x_2}(x)\doteq \nabla_{x_2}H(x).\)
This class of systems has been studied by~\cite{Schaft2021,Schaft2022},
where block-diagonal port-Hamiltonian structures are shown to
characterize a broad class of energy-conversion systems and to imply
one-port cyclo-passivity properties. In contrast to interconnected
multi-domain systems, where distinct components exchange energy through
interconnection ports, the present class describes a single
energy-conversion component in which storage variables from different
physical domains are coupled through a common Hamiltonian.
Representative examples include electromechanical actuators,
thermo-mechanical systems such as gas--piston systems, and
electro-thermal devices; see \citep{Schaft2021,Schaft2022}. In contrast, systems formed by the
interconnection of separate electrical, mechanical, or thermal
subsystems through transformers, gyrators, or more general Dirac
interconnections do not belong to this class.

As illustrated in Fig.~\ref{fig:2port_diagnoal}, the two external ports
are geometrically decoupled in the interconnection structure but interact
through the common Hamiltonian $H(x_1,x_2)$.
Consequently, the storage variation associated with $x_1$ and $x_2$ can
be related to the corresponding port power flows as follows.

\begin{figure}
    \centering
    \includegraphics[height=0.4\linewidth]{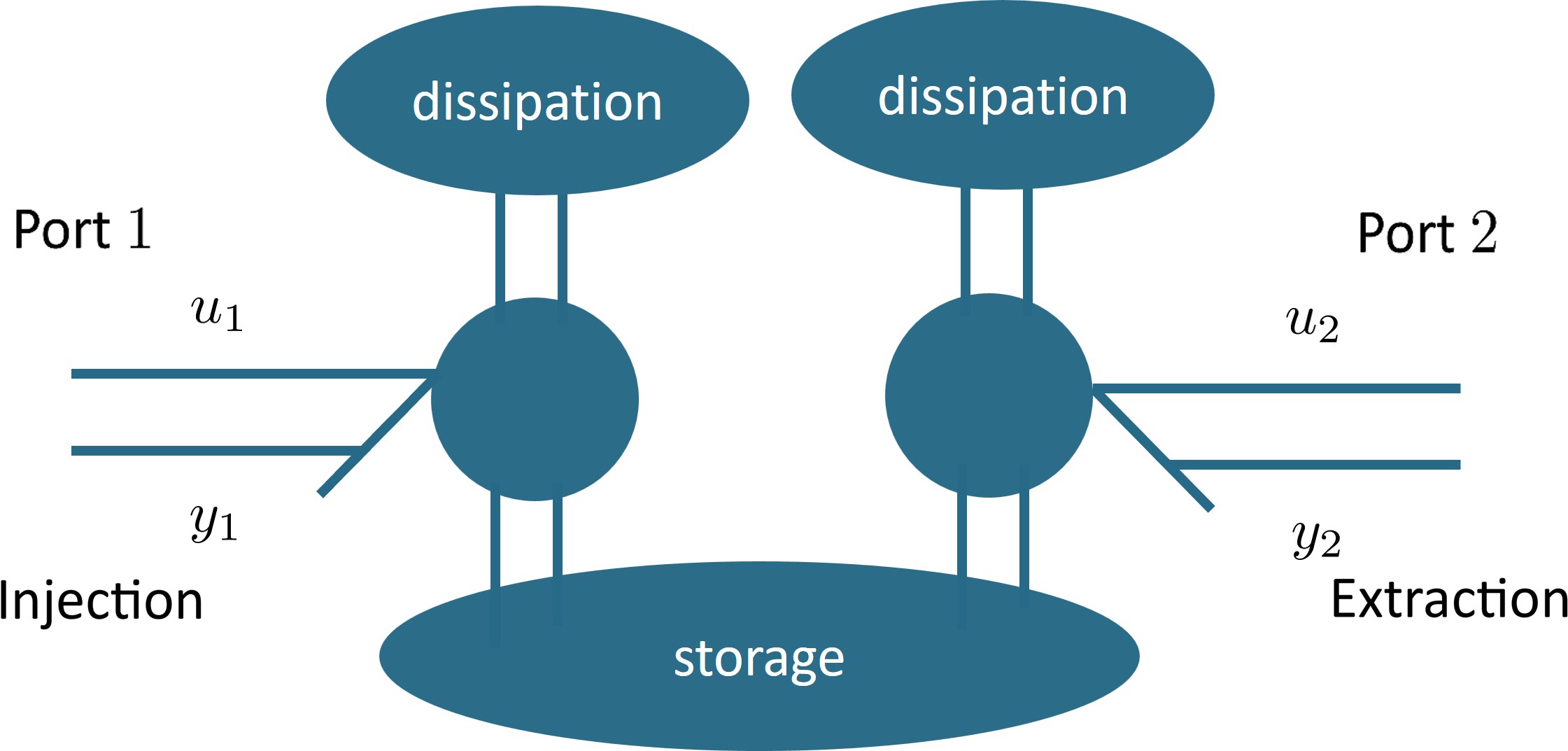}
    \caption{Two-port pH system in block-diagonal form}
    \label{fig:2port_diagnoal}
\end{figure}

\begin{lemma}[Subsystem power inequalities]\label{lem:power_bounds}~\\
Consider the block-diagonal port-Hamiltonian system~\eqref{eq:PH_block}.
Along any system trajectory, we have:
\begin{subequations}\label{eq:upperbound}
\begin{align}
H_{x_1}^\top(x)\dot x_1 &\le y_1^\top u_1, \label{eq:u1y1}\\
H_{x_2}^\top(x)\dot x_2 &\le y_2^\top u_2. \label{eq:u2y2}
\end{align}
\end{subequations} 
\end{lemma}
\begin{pf}
    Multiplying the $x_1$- and $x_2$-subsystems of \eqref{eq:PH_block} by $H_{x_1}^\top$ and $H_{x_2}^\top$,
respectively, and using the skew-symmetry of $J_i$ and positive semidefiniteness of $R_i$,
yields
\begin{subequations}\label{eq:passivity_subsystems}
\begin{align}
H_{x_1}^\top\dot x_1
&=H_{x_1}^\top\big(J_1-R_1\big)H_{x_1} + H_{x_1}^\top G_1u_1 \nonumber\\
&=
y_1^\top u_1 - H_{x_1}^\top R_1H_{x_1}
\ \le\ y_1^\top u_1, \\[1mm]
H_{x_2}^\top\dot x_2
&=H_{x_2}^\top\big(J_2(x)-R_2(x)\big)H_{x_2} + H_{x_2}^\top G_2u_2 \nonumber\\
&=
y_2^\top u_2 - H_{x_2}^\top(x)R_2(x)H_{x_2}(x)
\ \le\ y_2^\top u_2. 
\end{align} 
\end{subequations}
\End
\end{pf}

The relations in~\eqref{eq:upperbound} connect the storage variations of
the two subsystems with the power flows at the corresponding ports.
Moreover, the block-diagonal form~\eqref{eq:PH_block} imposes structural
limitations on the steady-state energy conversion.
\begin{lemma}[No energy conversion at steady state]\label{lem:limits}~\\
Consider system~\eqref{eq:PH_block}. For any steady-state operating point $(\bar u,\bar x,\bar y)$, the following identities hold:
\begin{subequations}\label{eq:passity_subsystem_ss}
\begin{align}
0 &= \bar y_1^\top \bar u_1 
    - H_{x_1}^\top(\bar x)\, R_1(\bar x)\, H_{x_1}(\bar x), \\
0 &= \bar y_2^\top \bar u_2 
    - H_{x_2}^\top(\bar x)\, R_2(\bar x)\, H_{x_2}(\bar x).
\label{eq:port2_nonneg_int}
\end{align}
\end{subequations}
\end{lemma}
\begin{pf} Equation
\eqref{eq:passity_subsystem_ss} follows by evaluating~\eqref{eq:passivity_subsystems} at equilibriums. \End
\end{pf}

Lemma~\ref{lem:limits} shows that nontrivial steady-state energy
conversion is impossible. By~\eqref{eq:port2_nonneg_int}, the net power
flow through Port~2 is nonnegative at steady state, implying that no net
energy can be extracted. Hence, nontrivial energy conversion
necessarily requires time-varying operation.

For cyclic operation, the Hamiltonian returns to its initial value after
each cycle, i.e., $H(x(\Tp))-H(x(0))=0$. Consequently, the storage
variations associated with $x_1$ and $x_2$ balance over one cycle.
This relation yields an upper bound on the energy that can be extracted
through Port~2, as stated in the following lemma.
\begin{lemma}[Extracted energy upper bound]\label{lem:p2_upperbound}
Consider system~\eqref{eq:PH_block}. For any cyclic operation over $[0,\Tp]$
satisfying $x(0)=x(\Tp)$, it holds that
\begin{equation}\label{eq:power_upperbound}
\int_0^{\Tp} -u_2^\top y_2 \,\mathrm dt
\;\le\;
\int_0^{\Tp} H_{x_1}^\top(x)\,\dot x_1\,\mathrm d t .
\end{equation}
\end{lemma}

\begin{pf}
Along any trajectory, the Hamiltonian satisfies
\[
\dot H(x)=H_{x_1}^\top(x)\,\dot x_1+H_{x_2}^\top(x)\,\dot x_2.
\]
Integrating over $[0,\Tp]$ and using the cyclic condition $x(0)=x(\Tp)$ yields
\[
0
=
\int_0^{\Tp}\Big(H_{x_1}^\top(x)\,\dot x_1+H_{x_2}^\top(x)\,\dot x_2\Big)\,\mathrm dt.
\]
Hence,
\begin{equation}\label{eq:Hx2dot_relation}
-\int_0^{\Tp} H_{x_2}^\top(x)\,\dot x_2\,\mathrm dt
=
\int_0^{\Tp} H_{x_1}^\top(x)\,\dot x_1\,\mathrm dt.
\end{equation}
On the other hand, by~\eqref{eq:u2y2},
\[
H_{x_2}^\top(x)\,\dot x_2
=
y_2^\top u_2 - H_{x_2}^\top(x)\,R_2(x)\,H_{x_2}(x),
\]
which implies
\[
-y_2^\top u_2
=
- H_{x_2}^\top(x)\,\dot x_2 - H_{x_2}^\top(x)\,R_2(x)\,H_{x_2}(x).
\]
Integrating over $[0,\Tp]$ and using~\eqref{eq:Hx2dot_relation} gives
\begin{align*}
\int_0^{\Tp} -y_2^\top u_2 \,\mathrm dt
&=
\int_0^{\Tp} H_{x_1}^\top(x)\,\dot x_1\,\mathrm dt \\
&\quad - \int_0^{\Tp} H_{x_2}^\top(x)\,R_2(x)\,H_{x_2}(x)\,\mathrm dt.
\end{align*}
Since $R_2(x)\succeq 0$, the second term is nonnegative, which proves~\eqref{eq:power_upperbound}. \End
\end{pf}

The following proposition extends
Lemma~\ref{lem:p2_upperbound} from cyclic trajectories to the
finite-horizon optimal trajectory.

\begin{proposition}
~\label{pro:p2_upper_noncyc}
Consider OCP~\eqref{eq:OCP_Efficiency} with dynamics in the block-diagonal
form~\eqref{eq:PH_block}. Let its optimal trajectory be$(u^\star(\cdot),y^\star(\cdot),x^\star(\cdot))$. Then for all $x_0 \in \mathbb X_0$,
\[
\int_0^{\Tf} -u^{\star\top}_2 y^\star_2 \,\mathrm d t
\le
\int_0^{\Tf} H_{x_1}^\top(x^\star)\,\dot x_1^\star\,\mathrm dt + \bar H .
\]    
\end{proposition}

\begin{pf}
Since the optimal trajectory is not necessarily cyclic, instead of
\eqref{eq:Hx2dot_relation} we obtain from the Hamiltonian balance
\begin{multline*}
-\int_0^{{\Tf}} H_{x_2}^\top(x^\star)\,\dot x_2^\star\, \mathrm  dt
=
\int_0^{{\Tf}} H_{x_1}^\top(x^\star)\,\dot x_1^\star\,\mathrm  dt \\
+ H(x^\star(0)) - H(x^\star({\Tf})).
\end{multline*}
Using the storage bound~\eqref{eq:barH} gives
\[
-\int_0^{{\Tf}} H_{x_2}^\top(x^\star)\,\dot x_2^\star\,\mathrm  dt
\le
\int_0^{{\Tf}} H_{x_1}^\top(x^\star)\,\dot x_1^\star\,\mathrm  dt + \bar H .
\]
The remainder of the argument follows exactly the same steps as in the
proof of Lemma~\ref{lem:p2_upperbound}. \End
\end{pf}

\section{Generalized Carnot Cycle and Carnot efficiency bound}~\label{sec:carnot}
The following result shows that, under an additional structural
restriction on system~\eqref{eq:PH_block}, the achievable efficiency~${\eta_{+}}$ is fundamentally limited by a Carnot-type efficiency bound.

\begin{assumption}[Scalar storage coordinate]
\label{ass:scalar_storage_coordinate}
Let
$
\Pi_1(x)=x_1
$
denote the projection onto the $x_1$-coordinates. Suppose that there exists a continuously differentiable function
$c: \Pi_1(\mathbb X)\to\mathbb R$ and a continuously differentiable function
$\widetilde H$ such that
\[
H(x_1,x_2)=\widetilde H(c(x_1),x_2).
\]
Define
\[
H_c(c,x_2)
\doteq
\nabla_c\widetilde H(c,x_2).
\]
Assume that there exist constants
$\bar H_c\geq\underline H_c\ge 0$ such that
\[
0\le \underline H_c
\le
H_c(c(x_1),x_2)
\le
\overline H_c,
\qquad
\forall x\in\mathbb X. \tag*{\End}
\]
\end{assumption}

Under Assumption~\ref{ass:scalar_storage_coordinate}, one can construct a
generalized Carnot cycle in the $H_c$--$c$ plane by extending the
construction by~\cite{Schaft2021,Schaft2022}. The resulting cycle serves as the reference process for deriving the fundamental upper bound on the maximal achievable efficiency.

As depicted in Figure~\ref{fig:Carnot}, the generalized Carnot cycle
consists of the following four phases:
\begin{itemize}
    \item[a)] Isothermal-like phase on $[0,\tau_1]$: design inputs $u_1$ and $u_2$ such that $H_c=\overline H_c$ and the system evolves from $c(0)=\underline c$ to $c(\tau_1)=\overline c$.
    
    \item[b)] Adiabatic-like phase on $[\tau_1,\tau_2]$: design inputs $u_1$ and $u_2$ such that $c=\overline c$ and the system evolves from $H_c(\tau_1)=\overline H_c$ to $H_c(\tau_2)=\underline H_c$.
    
    \item[c)] Isothermal-like phase on $[\tau_2,\tau_3]$: design inputs $u_1$ and $u_2$ such that $H_c=\underline H_c$ and the system evolves from $c(\tau_2)=\overline c$ to $c(\tau_3)=\underline c$.
    
    \item[d)] Adiabatic-like phase on $[\tau_3,\Tp]$: design inputs $u_1$ and $u_2$ such that $c=\underline c$ and the system evolves from $H_c(\tau_3)=\underline H_c$ to $H_c(\Tp)=\overline H_c$.
\end{itemize}

\begin{figure}
    \centering
    \includegraphics[height=0.4\linewidth]{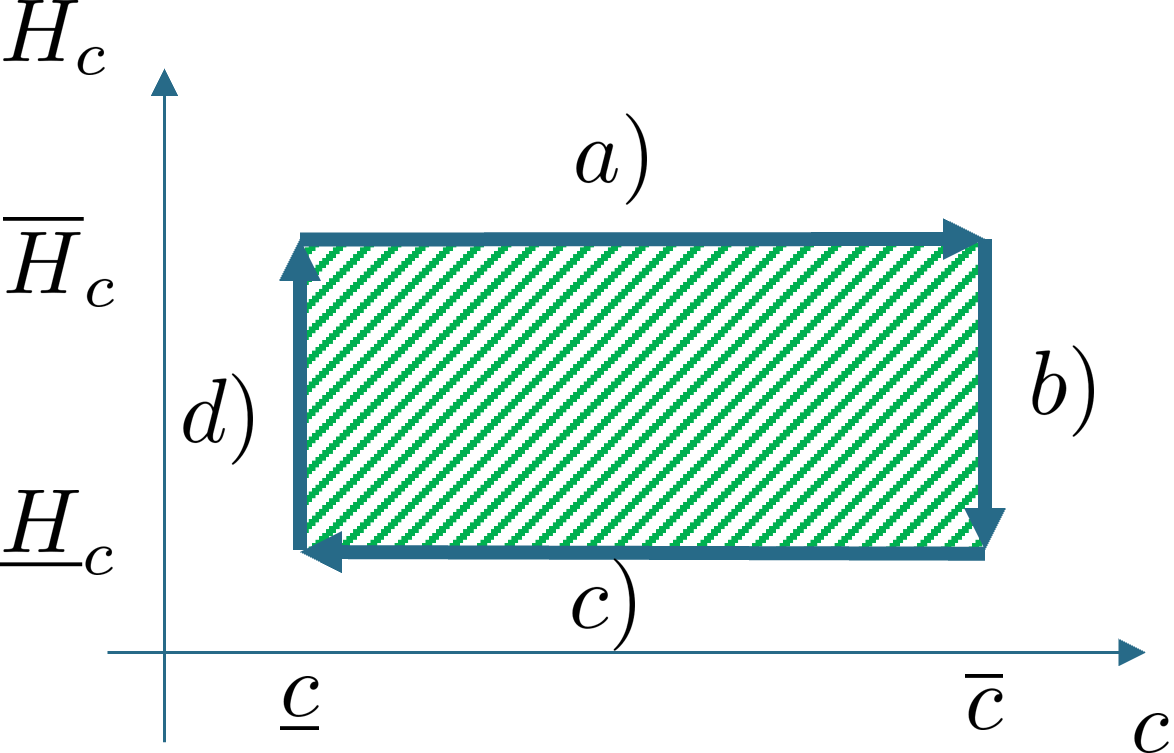}
    \caption{A generalized Carnot cycle in the $H_c$--$c$ plane.}
    \label{fig:Carnot}
\end{figure}

The above construction is formally analogous to the classical
thermodynamic Carnot cycle, where $H_c$ and $c$ play the roles of
generalized co-energy and energy variables analogous to temperature and
entropy, respectively. The isothermal-like phases correspond to
trajectories with constant co-energy variable $H_c$, whereas the
adiabatic-like phases correspond to trajectories with constant storage
coordinate $c$. In the absence of dissipation, the extracted energy of the generalized
Carnot cycle satisfies
\[
\oint u_2^\top y_2\,\mathrm dt
=
\oint u_1^\top y_1\,\mathrm dt
=
\oint H_c\,\mathrm dc
=
(\overline c-\underline c)
(\overline H_c-\underline H_c).
\]

The positive injected energy satisfies
\[
\oint [u_1^\top y_1]_+\,\mathrm dt
=
\oint H_c(t)[\mathrm dc(t)]_+
=
(\overline c-\underline c)\overline H_c.
\]

Taking the ratio of the extracted energy to the positive injected
energy yields the generalized Carnot efficiency~\citep{Struchtrup_2024},
\begin{equation}\label{eq:carnot}
\eta_{\mathrm{carnot}}
\doteq
1-
\frac{\underline H_c}{\overline H_c},
\end{equation}
which reduces to the classical Carnot efficiency when $H_c$ and $c$
represent temperature and entropy, respectively.

The generalized Carnot efficiency characterizes the best achievable
energy-conversion performance under the structural restriction of
Assumption~\ref{ass:scalar_storage_coordinate}. The following theorem shows that no
feasible cyclic trajectory can exceed this limit.

\begin{theorem}[Carnot efficiency bound]
\label{thm:carnot}
Consider $\mathrm{OCP}$~\eqref{eq:OCP_clc_Eff} with dynamics in the
block-diagonal form~\eqref{eq:PH_block} with further structural restriction by Assumption~\ref{ass:scalar_storage_coordinate}. Then, for every
$\Tp\in\mathbb R^+$ and every feasible cyclic trajectory
$(x(\cdot),u(\cdot))$ satisfying
\[
\int_0^{\Tp}[u_1^\top(t)y_1(t)]_+\,\mathrm dt>0,
\]
it holds that
\begin{equation}\label{eq:carnot_inequality}
{\eta_{+}}(x_0,u(\cdot))
\le
\eta_{\mathrm{carnot}}.
\end{equation}
\end{theorem}

\begin{pf}
By Assumption~\ref{ass:scalar_storage_coordinate},
\[
H(x_1,x_2)=\widetilde H(c(x_1),x_2).
\]
Hence, the chain rule gives
\[
H_{x_1}(x)
=
H_c(c(x_1),x_2)\nabla_{x_1} c(x_1),
\]
and therefore
\begin{equation}\label{eq:scalar_storage_identity}
H_{x_1}^\top\dot x_1
=
H_c \left(\nabla_{x_1} c(x_1)\right)^\top\dot x_1
=
H_c\dot c.
\end{equation}

Consider any feasible cyclic trajectory of
$\mathrm{OCP}$~\eqref{eq:OCP_clc_Eff}. By
Lemma~\ref{lem:p2_upperbound} and
\eqref{eq:scalar_storage_identity},
\[
-\int_0^{\Tp}u_2^\top(t)y_2(t)\,\mathrm dt
\le
\int_0^{\Tp}H_c(t)\dot c(t)\,\mathrm dt.
\]
Moreover, by~\eqref{eq:u1y1},
\[
H_c(t)\dot c(t)
=
H_{x_1}^\top(x(t))\dot x_1(t)
\le
u_1^\top(t)y_1(t),
\]
and thus
\[
[H_c(t)\dot c(t)]_+
\le
[u_1^\top(t)y_1(t)]_+.
\]
Consequently,
\begin{align*}
{\eta_{+}}(x_0,u(\cdot))
&=
\frac{-\int_0^{\Tp}u_2^\top(t)y_2(t)\,\mathrm dt}
{\int_0^{\Tp}[u_1^\top(t)y_1(t)]_+\,\mathrm dt}
\le
\frac{\int_0^{\Tp}H_c(t)\dot c(t)\,\mathrm dt}
{\int_0^{\Tp}[H_c(t)\dot c(t)]_+\,\mathrm dt}.
\end{align*}
Using $z=[z]_+-[-z]_+$, the right-hand side becomes
\[
1-
\frac{\int_0^{\Tp}[-H_c(t)\dot c(t)]_+\,\mathrm dt}
{\int_0^{\Tp}[H_c(t)\dot c(t)]_+\,\mathrm dt}.
\]
Since $H_c(t)\ge0$, the bounds on $H_c$ imply
\begin{equation}\label{eq:thermal_upper}
{\eta_{+}}(x_0,u(\cdot))
\le
1-
\frac{
\underline H_c
\int_0^{\Tp}[-\dot c(t)]_+\,\mathrm dt}
{
\overline H_c
\int_0^{\Tp}[\dot c(t)]_+\,\mathrm dt}.
\end{equation}

Periodicity  of $x_1$ implies
\[
c(x_1(\Tp))=c(x_1(0)),
\]
and hence
\begin{equation}\label{eq:equality_c}
0
=
\int_0^{\Tp}\dot c(t)\,\mathrm dt
=
\int_0^{\Tp}[\dot c(t)]_+\,\mathrm dt
-
\int_0^{\Tp}[-\dot c(t)]_+\,\mathrm dt.
\end{equation}
Therefore,
\[
\int_0^{\Tp}[-\dot c(t)]_+\,\mathrm dt
=
\int_0^{\Tp}[\dot c(t)]_+\,\mathrm dt.
\]
Substitution into~\eqref{eq:thermal_upper} yields
\[
{\eta_{+}}(x_0,u(\cdot))
\le
1-\frac{\underline H_c}{\overline H_c}
=
{\eta_{+}}_{\mathrm{carnot}}. \tag*{\End}
\]
\end{pf}

Although Theorem~\ref{thm:carnot} applies directly only to cyclic
operation, the same structural argument yields an asymptotic upper bound
for finite-horizon optimal trajectories, which need not be periodic.

\begin{proposition}[Asymptotic Carnot-type upper bound] \hfill
\label{pro:carnot}
Let Assumption~\ref{ass:scalar_storage_coordinate} hold.
Consider OCP~\eqref{eq:OCP_Efficiency} with dynamics in the
block-diagonal form~\eqref{eq:PH_block}. Then, for every
$x_0\in\mathbb X_0$,
\[
\limsup_{\Tf\to\infty}\eta_{+}^\star(x_0,\Tf)
\le
\eta_{\mathrm{carnot}}.
\]
\end{proposition}

\begin{pf}
Let $(x^\star(\cdot),u^\star(\cdot),y^\star(\cdot))$ be an optimal
trajectory of OCP~\eqref{eq:OCP_Efficiency}. By the  same argument as in the proof
of Theorem~\ref{thm:carnot},
\begin{equation}\label{eq:thermal_upper_noncyc}
\eta_{+}^\star(x_0,\Tf)
\le
1-
\frac{
\underline H_c
\int_0^{\Tf}[-\dot c^\star(t)]_+\,\mathrm dt}
{
\overline H_c
\int_0^{\Tf}[\dot c^\star(t)]_+\,\mathrm dt}.
\end{equation}
Here and below,
$c^\star(t)\doteq c(x_1^\star(t))$.

However, since the trajectory is not necessarily cyclic,
\eqref{eq:equality_c} no longer holds. Instead,
\[
\int_0^{\Tf}[-\dot c^\star(t)]_+\,\mathrm dt
\ge
\int_0^{\Tf}[\dot c^\star(t)]_+\,\mathrm dt
-
|c^\star(\Tf)-c^\star(0)|.
\]
Let
$
\Pi_1(x)=x_1
$
denote the projection onto the $x_1$-coordinates. Since $\mathbb X$ is
compact, $\Pi_1(\mathbb X)$ is compact as the continuous image of
$\mathbb X$. Hence, by continuity of $c$, there exists
$\Delta_c<\infty$ such that
\[
|c(x_1)-c(\tilde x_1)|
\le
\Delta_c,
\qquad
\forall\,x_1,\tilde x_1\in\Pi_1(\mathbb X).
\]
Substitution into~\eqref{eq:thermal_upper_noncyc} yields
\[
\eta_{+}^\star(x_0,\Tf)
\le
\eta_{\mathrm{carnot}}
+
\frac{\underline H_c\Delta_c}
{\overline H_c
\int_0^{\Tf}[\dot c^\star(t)]_+\,\mathrm dt}.
\]

It remains to show that the denominator diverges. By
Proposition~\ref{pro:p2_upper_noncyc} and
\eqref{eq:scalar_storage_identity},
\[
E_2^\star(x_0,\Tf)
\le
\int_0^{\Tf}H_c(t)\dot c^\star(t)\,\mathrm dt+\overline H
\le
\overline H_c
\int_0^{\Tf}[\dot c^\star(t)]_+\,\mathrm dt+\overline H.
\]
Together with the extraction constraint
$E_2^\star(x_0,\Tf)\ge p\Tf$, this gives
\[
\lim_{T\to \infty}\overline H_c
\int_0^{\Tf}[\dot c^\star(t)]_+\,\mathrm dt
\ge \lim_{T\to \infty}
p\Tf-\overline H =\infty.
\]
Hence, the correction term vanishes as $\Tf\to\infty$, and therefore
\[
\limsup_{\Tf\to\infty}\eta_{+}^\star(x_0,\Tf)
\le
\eta_{\mathrm{carnot}}. \tag*{\End}
\]
\end{pf}
The scalar storage coordinate in Assumption~\ref{ass:scalar_storage_coordinate}
is essential for the proof of Theorem~\ref{thm:carnot}. In particular, the
derivation of~\eqref{eq:thermal_upper} relies on the scalar identity
\[
[H_c\dot c]_+ = H_c[\dot c]_+,\qquad H_c>0,
\]
which holds only for a scalar storage coordinate.

In contrast, for a genuinely multi-dimensional storage coordinate
$c\in\mathbb R^n$, no analogous identity holds for
$H_c^\top\dot c$, since the components of $\dot c$ may have different
signs and cancel each other in the inner product. As a result, the
above proof cannot be extended directly. Nevertheless, the optimal
cyclic problem~\eqref{eq:OCP_clc_Eff_sup} remains well posed for general
block-diagonal port-Hamiltonian systems and, by
Theorem~\ref{thm:asymptotic_upper}, still provides a tight asymptotic
upper bound on the achievable conversion efficiency.

\begin{remark}[Relation to the classical Carnot cycle]~\\
\label{rem:carnot}
The bound in Theorem~\ref{thm:carnot} recovers the classical Carnot
efficiency as a special case. For a thermodynamic system, let $c=S$
denote entropy. By the Gibbs relation,
$
    H_c=\nabla_S H=\Theta,
$
with $\Theta$ denoted as temperature. Hence, the bound on $\eta_+$ reduces to
\[
    \eta_+\leq 1-\frac{\underline \Theta}{\overline \Theta},
\]
i.e., the classical Carnot bound.

However, there is a difference in the interpretation of the
corresponding cycle. The classical Carnot cycle considers heat exchange
with hot and cold reservoirs at fixed temperatures. Reversible heat
transfer is understood in the quasi-static limit for  isothermal phases, when the temperature
 of the  system and the respective reservoir
remains the same. In the present control formulation,  isothermal phase only require to keep the
temperature of the system constant; see also
Remark~2 in~\cite{Schaft2021a}. This allows the generalized Carnot
cycle to be considered as a dynamic process over a finite period, as
illustrated by the gas-piston example in Section~\ref{sec:simu_carnot}.
\end{remark}

\section{Simulation Examples}\label{sec:simulation}
The previous sections establish the relation between long-horizon and cyclic optimal operation and derive a Carnot-type efficiency bound under additional structural conditions. We illustrate these results using two complementary examples. First, a nonlinear reversible gas-piston system satisfying Assumption~\ref{ass:scalar_storage_coordinate} demonstrates that cyclic optimization of \({\eta_{+}}\) recovers the classical Carnot cycle and attains the corresponding Carnot bound. Second, a dissipative linear port-Hamiltonian system that does not satisfy Assumption~\ref{ass:scalar_storage_coordinate}  illustrates the general asymptotic relation between finite-horizon and optimal cyclic efficiencies. The OCPs are implemented in MATLAB R2025a using CasADi 3.7.2 and solved with IPOPT.

\subsection{A reversible lossless gas-piston system}~\label{sec:simu_carnot}

In this example, we consider a reversible, lossless gas–piston system. We show that optimizing the conversion efficiency \({\eta_{+}}\) over cyclic operations recovers the Carnot cycle. In contrast, since the system is lossless, the efficiency \(\eta\) equals one for any cyclic operation and therefore does not distinguish between different cyclic operating regimes.

 Consider the following gas-piston system with incoming entropy flow
\begin{subequations}\label{eq:gas_piston}
    \begin{align}
    \begin{bmatrix}
        \dot S\\
        \dot V\\
        \dot \pi
    \end{bmatrix} &= \begin{bmatrix}
        0 & 0 & 0\\
        0 &0 &A \\
        0 &- A & 0
    \end{bmatrix} \begin{bmatrix}
        \nabla_S H\\
        \nabla_V H\\
        \nabla_\pi H
    \end{bmatrix} + \begin{bmatrix}
        1& 0\\
        0& 0 \\
        0 & 1
    \end{bmatrix} \begin{bmatrix}
        u_1\\
        u_2
    \end{bmatrix}
\end{align}
where $S$ is the entropy of the system, $V$ is the volume of the gas,
$\pi$ is the momentum of the piston. The Hamiltonian $H$ consists of the internal energy of the gas and the kinetic energy of the piston
\begin{equation}
     H(S,V,\pi) = U(S,V) + \frac{\pi^2}{2m},\,
\end{equation}
with 
\[
 U(S,V) = U_0 e^{\frac{2S}{3 RN}}  V^{-\frac{2}{3}},\,U_0
=
\frac{3}{2}NR\temp_0
e^{-\frac{2s_0}{3R}}
\left(\frac{NR\temp_0}{P_0}\right)^{2/3}
\]
with the following properties
 \begin{align}
     \nabla_\pi H &= \frac{\pi}{m} \doteq v\\
    \nabla_S H &= \nabla_S U  = \frac{2 U}{3 N R
    } =  \frac{2U_0}{3 N R} e^{\frac{2S}{3 RN}}  V^{-\frac{2}{3}} \doteq \temp\label{eq:temperature}\\
\nabla_V H &=\nabla_V U= -\frac{2U }{3 V} = - \frac{2}{3}U_0 e^{\frac{2S}{3 RN}}  V^{-\frac{5}{3}}\doteq -P\\
PV &= NR \temp
 \end{align}
 where $v$, $\temp$, and $P$ represent the velocity of the piston, the temperature, and the pressure of the system. The chosen parameters are summarized in Table~ \ref{tab:params}
 
 The outputs of the system are
 \begin{equation}   y_1 = \nabla_S H = \temp, \quad y_2 = \nabla_\pi H = \frac{\pi}{m}.
 \end{equation}
 \end{subequations}
This system is \textit{reversible and lossless} since it contains neither dissipation nor internal entropy production; entropy changes only through the external entropy-flow input. Irreversible thermodynamic processes require an extended pH formulation accounting explicitly for entropy production; see
~\cite{Ramirez2013,Philipp2024}.

\begin{table}[t]
\centering
\caption{Parameters and reference values.}
\label{tab:params}
\begin{tabular}{lcc}
\toprule
Parameter & Unit & Value \\
\midrule
Number of moles \(N\) & [--] & 0.01 \\
Reference pressure \(P_0\) & [Pa] & 101.325 \\
Reference temperature \(\temp_0\) & [K] & 273 \\
Reference molar entropy \(s_0\) & [kJ K\(^{-1}\) mol\(^{-1}\)] & 0.1100 \\
Mass \(m\) & [kg] & 0.5\\
Area \(A\) & [m\(^2\)] & 0.5 \\
\bottomrule
\end{tabular}
\end{table}

We impose the constraints
\begin{gather*}
0.01 \leq V \leq 10,\qquad
240 \leq \temp \leq 320,\qquad
-4 \leq \pi \leq 4,\\
-0.5 \leq u_1 \leq 0.5,\qquad
-5 \leq u_2 \leq 5.
\end{gather*}
The minimum energy extraction constraint is imposed with $p=0.5$.
In view of \eqref{eq:carnot}, the corresponding Carnot bound is
\[
\eta_{\mathrm{carnot}}
=
1-\frac{\underline \temp}{\overline \temp}
=
1-\frac{240}{320}
=
0.25.
\]

We do not impose direct bounds on the entropy. Instead, the volume is constrained to reflect the finite cylinder geometry. Since the entropy, temperature, and volume are related through the thermodynamic  relation~\eqref{eq:temperature}, the entropy is implicitly bounded by the imposed constraints on \(V\) and \(\temp\). 

The continuous-time system is discretized using the implicit midpoint method~\citep{Aoues2013} with sampling time \(T_s=0.05\,\mathrm{s}\). Since the Hamiltonian is nonlinear, the resulting discretization does not, in general, preserve the continuous-time energy balance exactly. 

We solve OCP~\eqref{eq:OCP_clc_Eff} for periods \(\Tp \in [2,10]\) with an increment of \(0.1\,\mathrm{s}\). As shown in Fig.~\ref{fig:effciency_nonlinear}, the cyclic conversion
efficiency reaches the Carnot bound $\eta_{\mathrm{Carnot}} =0.25$ for $\Tp\geq 3.0\,\mathrm{s}$, which is consistant as Theorem~\ref{thm:carnot}. Within the considered period grid, the
bound is first attained at $\Tp=3.0\,\mathrm{s}$; hence, following the
convention in Assumption~\ref{ass:Tpstar}, we select $\Tp^\star=3.0\,\mathrm{s}$.

As anticipated in Remark~\ref{rem:carnot}, the optimal solution provides
a finite-period dynamic realization of the generalized Carnot cycle.
Indeed, Fig.~\ref{fig:phase_nonlinear} shows that the optimal solution
of OCP~\eqref{eq:OCP_clc_Eff} for $\Tp=3.0\,\mathrm{s}$ recovers the
cycle introduced in Section~\ref{sec:carnot} with
\[
x_1 = S, \qquad H_{x_1} = \temp,
\]
so that the isothermal-like phases correspond to constant temperatures $\temp = \overline \temp$ and $\temp = \underline \temp$, while the adiabatic-like phases correspond to constant entropies $S = \overline S$ and $S = \underline S$. 

\begin{figure}
    \centering
\includegraphics[width=0.75\columnwidth]{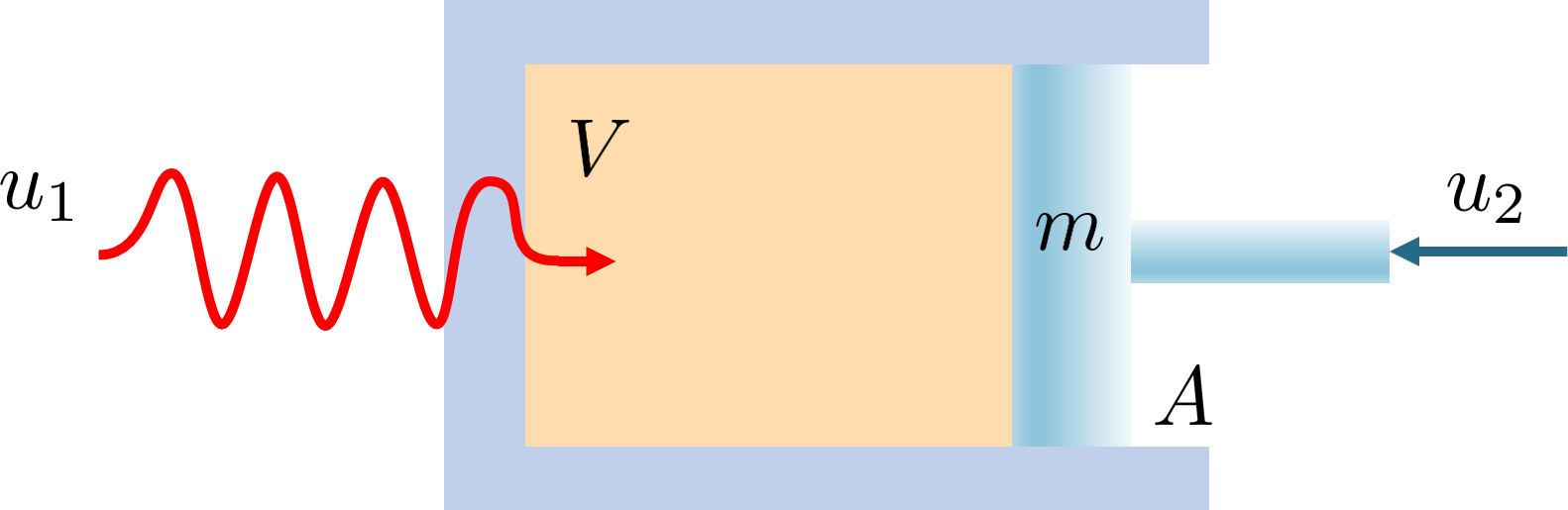}
    \caption{Schematic sketch of a gas-piston system.}
    \label{fig:gas_piston}
\end{figure}

\begin{figure}[t]
\centering
\includegraphics[width=1\columnwidth]{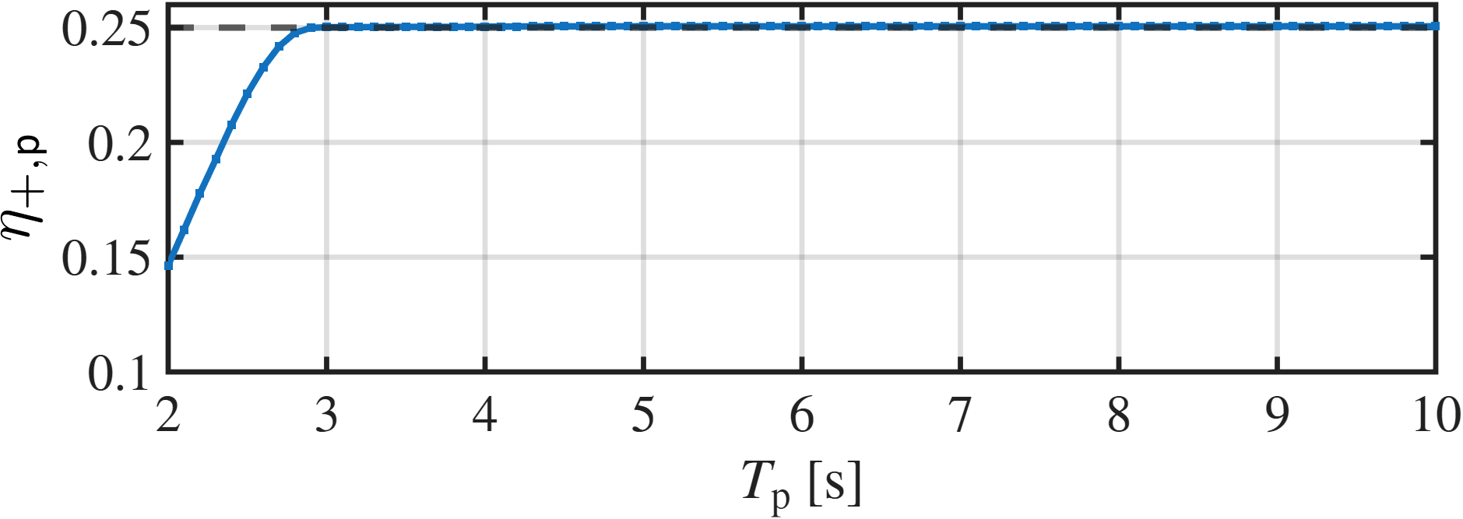}
\caption{Optimal cyclic conversion efficiency $\eta_{+,\text p} (T_p)$ for OCP~\eqref{eq:OCP_clc_Eff} with $\Tp \in [2,10]\, \mathrm{s}$ with system~\eqref{eq:gas_piston} (step size $0. 1 \, \mathrm{s}$). }
\label{fig:effciency_nonlinear}
\end{figure}

\begin{figure}
    \centering
\includegraphics[width=1\columnwidth]{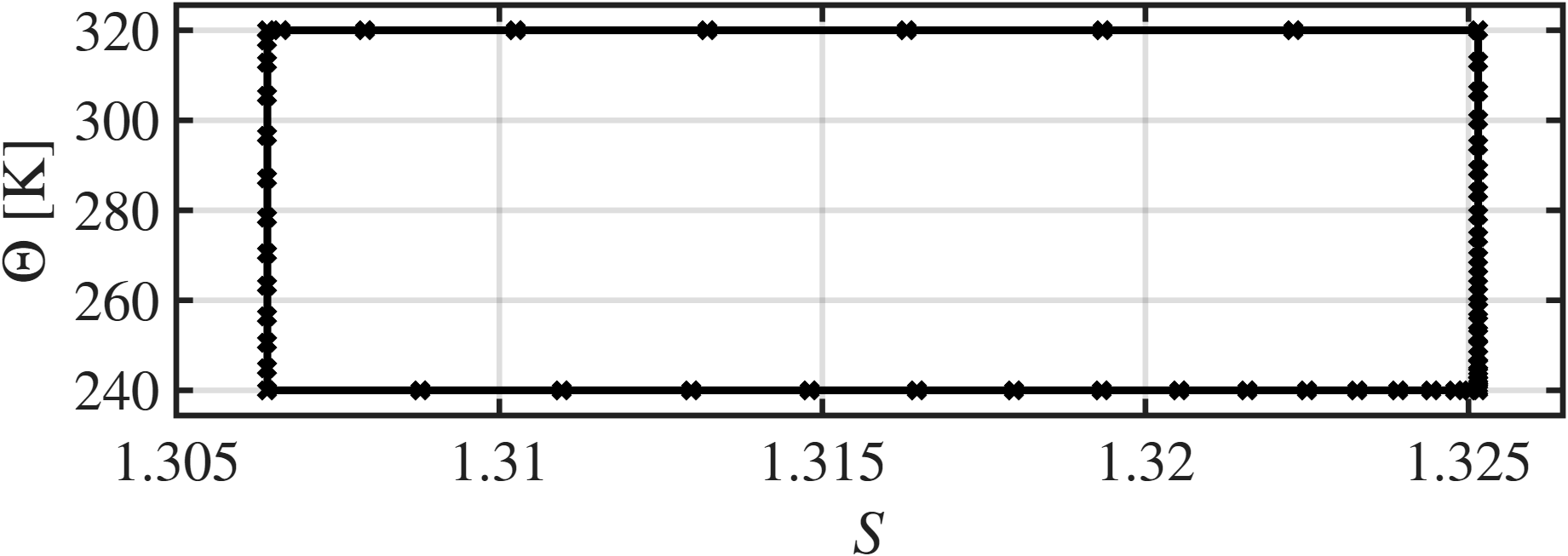}
    \caption{Solution to OCP~\eqref{eq:OCP_clc_Eff} with system~\eqref{eq:gas_piston} for $\Tp = 3.0\,\mathrm{s}$ in  $S-\Theta$ phase plane.}
    \label{fig:phase_nonlinear}
\end{figure}

\subsection{Linear port-Hamiltonian system in block-diagonal form}

We consider a linear port-Hamiltonian system with block-diagonal 
interconnection structure
\begin{equation*}
    \dot{x} = (JQ-R)x + Bu, 
    \qquad
    y = B^\top Qx,
\end{equation*}
where $R = 0.01 I_4$ and
\begin{align*}
J=
\begin{bmatrix}
\phantom{-}0  &1& \phantom{-}0 & 0\\
-1 & 0 & \phantom{-}0 & 0\\
\phantom{-}0 &0 &\phantom{-}0 & 1\\
\phantom{-}0 &0 & -1 & 0
\end{bmatrix}, \quad
Q=
\begin{bmatrix}
\phantom{-}1 & -1 & \phantom{-}0 & 0\\
-1 & \phantom{-}3 & -1 & 0\\
\phantom{-}0 & -1 & \phantom{-}1 &0\\
\phantom{-}0 & \phantom{-}0 &\phantom{-}0 & 1
\end{bmatrix}, \quad
B=
\begin{bmatrix}
0 &0\\
1 & 0\\
0 & 0\\
0 & 1
\end{bmatrix}.
\end{align*}
Let
\[
x =
\begin{bmatrix}
x_{1,1} & x_{1,2} & x_{2,1} & x_{2,2}
\end{bmatrix}^{\top},
\]
with
$x_1=[x_{1,1},x_{1,2}]^\top$ and
$x_2=[x_{2,1},x_{2,2}]^\top$.
The Hamiltonian
\[
H(x)=\frac{1}{2}x^\top Qx
\]
does not satisfy Assumption~\ref{ass:scalar_storage_coordinate}.
Hence, the Carnot-type bound of Theorem~\ref{thm:carnot} does not
apply to this example.

We impose
\begin{align*}
\begin{bmatrix}
-4\\-4
\end{bmatrix}
\leq x_1 \leq
\begin{bmatrix}
4\\4
\end{bmatrix},
\,
\begin{bmatrix}
-2\\-5
\end{bmatrix}
\leq x_2 \leq
\begin{bmatrix}
2\\5
\end{bmatrix},
\,
\begin{bmatrix}
-10\\-10
\end{bmatrix}
\leq u \leq
\begin{bmatrix}
10\\10
\end{bmatrix}.
\end{align*}
The minimum energy extraction constraint is imposed with $p=0.5$.

The continuous-time system is discretized using the implicit midpoint
method~\citep{Aoues2013} with sampling time $T_s=0.5\,\mathrm{s}$.
We first solve OCP~\eqref{eq:OCP_clc_Eff} for periods
$\Tp\in[2,20]\,\mathrm{s}$ with a step size of $0.5\,\mathrm{s}$.
As shown in Fig.~\ref{fig:effciency}, maximizing either $\eta$ or
${\eta_{+}}$ yields the same optimal efficiencies for all considered
periods. The maximum cyclic efficiency is attained at
$\Tp=8\,\mathrm{s}$, with
\[
\eta_\mathrm p(8)=\eta_{+,\mathrm p}(8)=0.9436.
\]

\begin{figure}[t]
    \centering
    \includegraphics[width=\columnwidth]{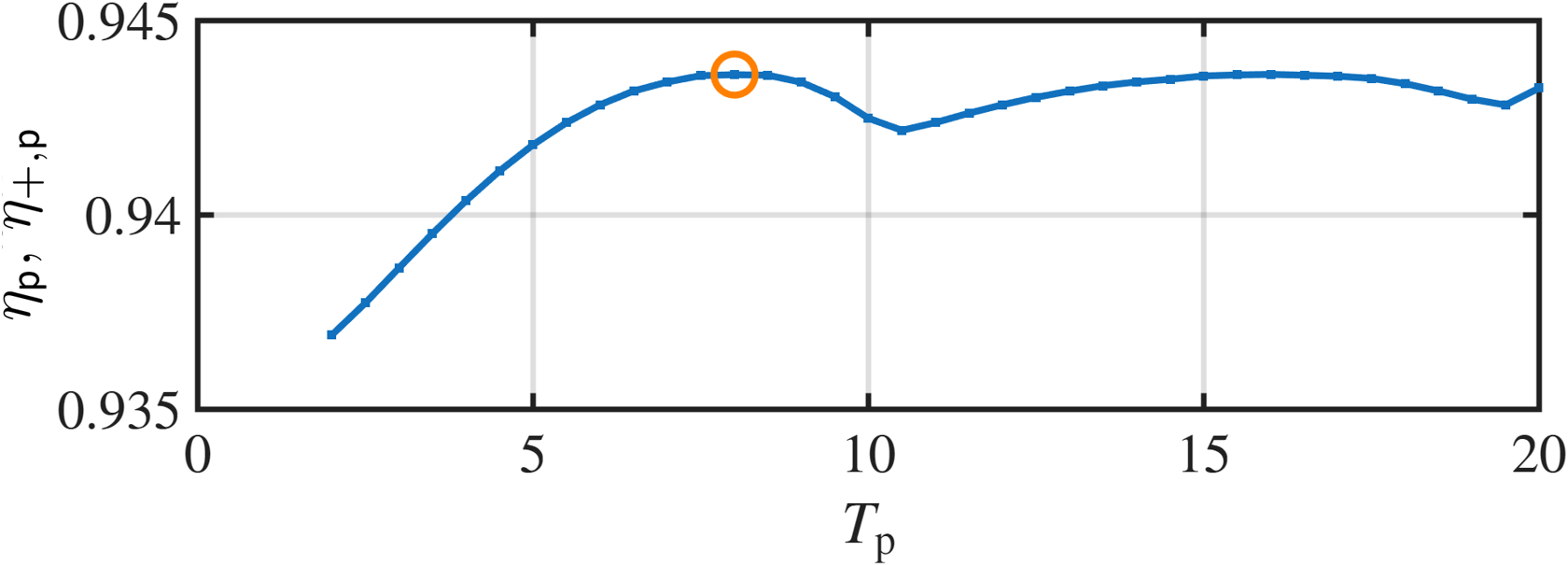}
    \caption{Optimal cyclic conversion efficiencies $\eta_p(\Tp)$ and
    $\eta_{+,\mathrm p}(\Tp)$ for OCP~\eqref{eq:OCP_clc_Eff} with
    $\Tp\in[2,20]\,\mathrm{s}$ and a step size of $0.5\,\mathrm{s}$.
    The circle marker indicates the maximum efficiency.}
    \label{fig:effciency}
\end{figure}

Next, we solve the finite-horizon OCP~\eqref{eq:OCP_Efficiency} for
\[
\Tf\in\{10,20,50,100,150,250\}\,\mathrm{s}
\]
with the fixed initial condition $x_0=[0,0,0,0]^\top$.
Motivated by the constructive argument in the proof of Theorem~\ref{thm:asymptotic_lower} , the finite-horizon OCP is warm-started by concatenating an initial transition with repetitions of the optimal cyclic trajectory.

The resulting optimal efficiencies are shown in
Fig.~\ref{fig:efficency_200}. As the horizon increases, both
efficiencies approach the best cyclic efficiency
$\eta_\mathrm p(8)=\eta_{+,\mathrm p(8)}=0.9436$. This numerical behavior is consistent
with the asymptotic results in Section~\ref{sec:asymptotic}: the
contribution of the initial and terminal transients becomes negligible
for increasing horizons, while cyclic operation determines the
long-horizon conversion performance.

\begin{figure}
    \centering
    \includegraphics[width=\linewidth]{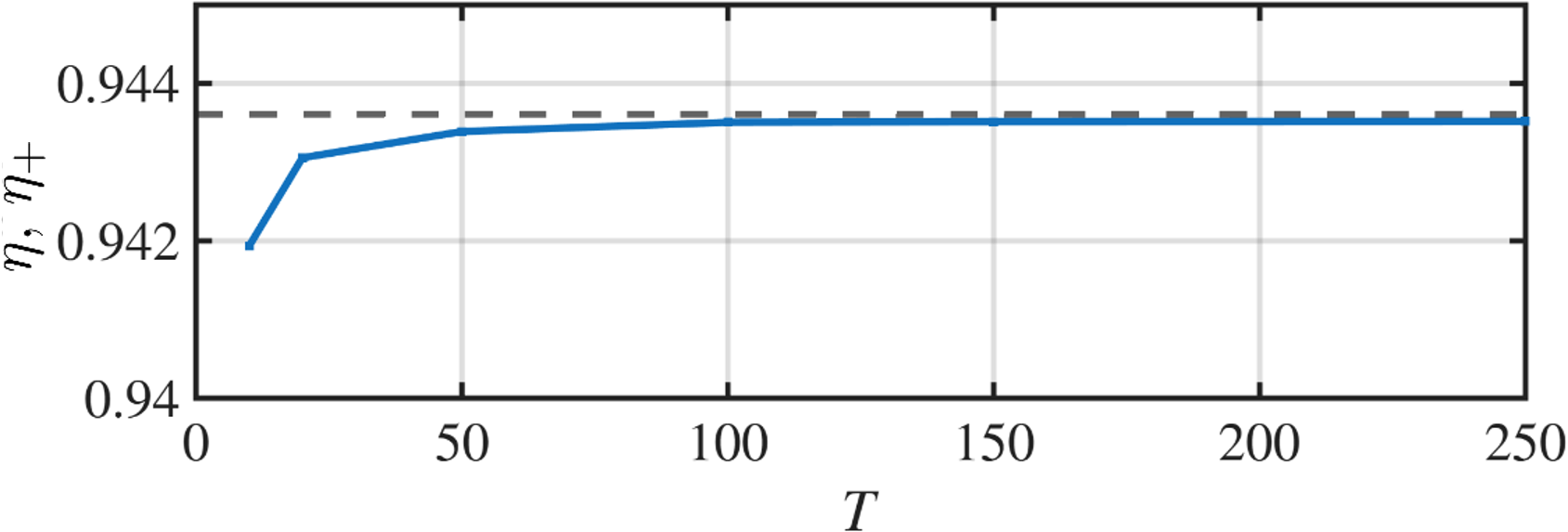}
    \caption{Optimal finite-horizon conversion efficiencies $\eta^\star$
    and $\eta_{+}^\star$ for OCP~\eqref{eq:OCP_Efficiency} with
    $\Tf\in\{10,20,50,100,150,250\}\,\mathrm{s}$. The dashed horizontal
    line indicates the best cyclic efficiency
    $\eta_\mathrm p=\eta_{+,\mathrm p}=0.9436$, attained for $\Tp=8\,\mathrm{s}$.}
    \label{fig:efficency_200}
\end{figure}

Figure~\ref{fig:input} shows the optimal port powers
$u_1^\top y_1$ and $-u_2^\top y_2$ for $\Tf=100\,\mathrm{s}$.
Along this optimal trajectory, $u_1^\top y_1\geq0$ for all
$t\in[0,100]$. Consequently,
\[
E_{1,+}=E_1,
\]
and the two efficiency measures coincide, i.e., $\eta={\eta_{+}}$.
In contrast, the extracted power $-u_2^\top y_2$ changes sign along
the trajectory, indicating that Port~2 alternates between energy
extraction and energy injection. Nevertheless, its net integral remains
positive and satisfies the prescribed minimum extraction constraint.

\begin{figure}
    \centering
    \includegraphics[width=\linewidth]{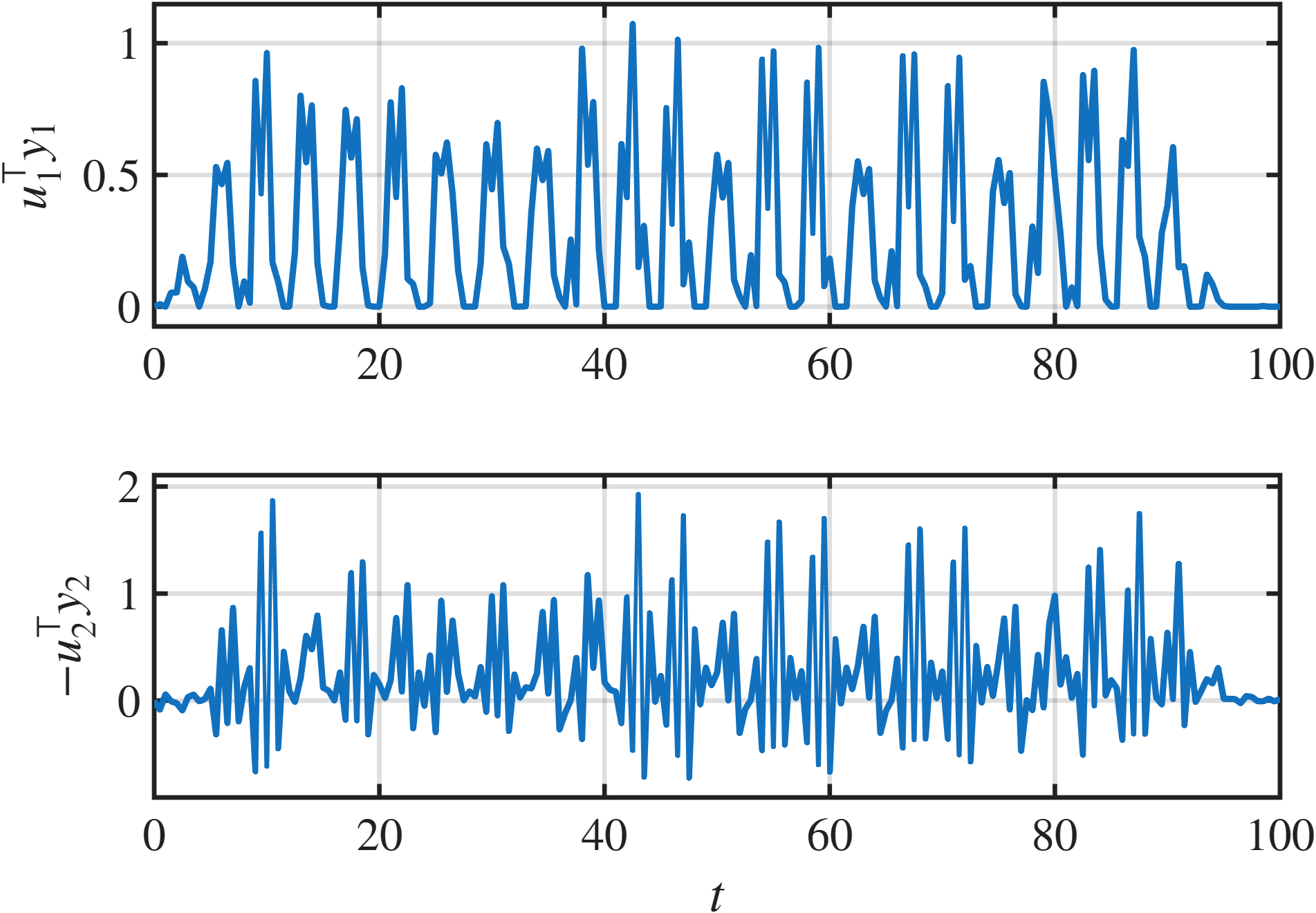}
    \caption{Optimal port powers $u_1^\top y_1$ and $-u_2^\top y_2$
    for OCP~\eqref{eq:OCP_Efficiency} with $\Tf=100\,\mathrm{s}$.}
    \label{fig:input}
\end{figure}

Although Assumption~\ref{ass:scalar_storage_coordinate} is not
satisfied and the Carnot-type bound of Section~\ref{sec:carnot} is
therefore not applicable, this example illustrates that the optimal
cyclic problem remains informative. In particular, the best cyclic
efficiency provides a tight characterization of the observed
long-horizon optimal conversion performance, consistent with the
general asymptotic results of Section~\ref{sec:asymptotic}.

\section{Conclusion}\label{sec:conclusion}

This work provides a control-theoretic perspective on energy conversion in port-Hamiltonian systems, revealing how long-horizon optimal behavior can be closely linked to cyclic operating regimes. In particular, the analysis shows that, under suitable controllability and reachability conditions, cyclic operation naturally emerges as an asymptotically optimal mechanism for achieving nontrivial energy conversion in structured systems. 

For the class of block-diagonal port-Hamiltonian systems introduced
by~\cite{Schaft2021,Schaft2022}, the system structure admits a more
refined interpretation. In particular, nontrivial energy conversion is
shown to be inherently dynamic, since steady-state operation cannot
sustain net energy extraction through Port~2. Furthermore, under an
additional scalar structural restriction, the generalized Carnot
efficiency constitutes a rigorous upper bound on the achievable
positive-injected conversion efficiency. This establishes a direct
connection between structural properties of port-Hamiltonian systems and
thermodynamic efficiency limits.

A promising direction for future research is the investigation of turnpike properties with respect to optimal cyclic orbits, which may provide a deeper understanding of the repeated emergence of optimal cyclic operation in long-horizon optimal trajectories. Finally, incorporating uncertainty and disturbances into the analysis could extend the framework to energy conversion problems arising in more realistic and uncertain environments.

\section*{Appendix}~\label{sec:Appendix}

We briefly discuss sufficient conditions ensuring the existence of optimal
controls for the finite-horizon OCP~\eqref{eq:OCP_Efficiency} 
as well as for the cyclic OCP~\eqref{eq:OCP_clc_Eff}. The argument is
based on the classical existence theorem by~\cite[Theorem~4, p. 259]{LeeMarkus1967}.

\begin{proposition}[Existence of optimal controls] 
\label{prop:existence_optimal_controls}
Let
\begin{itemize}
    \item[i)] the state and input constraint sets
    $\mathbb X\subset\mathbb R^{n_x}$ and
    $\mathbb U\subset\mathbb R^{n_u}$ be nonempty and compact, with
    $\mathbb U$ convex;

    \item[ii)] the mappings $J$, $R$, $B$, and $H_x$ be continuous on
    $\mathbb X$;

    \item[iii)] for the considered $x_0\in\mathbb X$ and $\Tf> S_a(x_0)/p$, OCP~\eqref{eq:OCP_Efficiency} be feasible 

    \item[iv)] for the considered $\Tp> S_a(x_0)/p$,
    OCP~\eqref{eq:OCP_clc_Eff} be feasible.
\end{itemize}
Then, for each $\phi\in\{\eta,{\eta_{+}}\}$, there exists an optimal control
for OCP~\eqref{eq:OCP_Efficiency}. Moreover, there exists an optimal cyclic
control for OCP~\eqref{eq:OCP_clc_Eff}.
\end{proposition}

\begin{pf}
Define
\[
f(x,u)
\doteq
\bigl(J(x)-R(x)\bigr)H_x(x)+B(x)u
\]
and introduce the accumulated extracted energy
\[
\dot z_2=-u_2^\top y_2,
\qquad
z_2(0)=0.
\]

For $\phi=\eta$, additionally introduce
\[
\dot z_1=u_1^\top y_1,
\qquad
z_1(0)=0.
\]
The objective is then represented by the terminal cost
\[
\eta=\frac{z_2(\Tf)}{z_1(\Tf)}.
\]
By i)--ii), the augmented dynamics are continuous and affine in $u$.
Hence, for each fixed $x\in\mathbb X$, the corresponding augmented
velocity set
\[
\left\{
\begin{bmatrix}
f(x,u)\\
u_1^\top y_1\\
-u_2^\top y_2
\end{bmatrix}
\,\middle|\,
u\in\mathbb U
\right\}
\]
is compact and convex.

For $\phi={\eta_{+}}$, introduce an auxiliary control $q$ and the accumulated
positive injected energy
\[
\dot z_+=q,
\qquad
z_+(0)=0,
\]
subject to
\[
0\le q\le p_1,
\qquad
q\ge u_1^\top y_1.
\]
Since $|u_1^\top y_1|\le p_1$, this constraint is feasible and the
corresponding augmented velocity set
\[
\left\{
\begin{bmatrix}
f(x,u)\\
-u_2^\top y_2\\
q
\end{bmatrix}
\,\middle|\,
u\in\mathbb U,\;
0\le q\le p_1,\;
q\ge u_1^\top y_1
\right\}
\]
is compact and convex for every fixed $x\in\mathbb X$.

This augmented formulation is equivalent to the original
${\eta_{+}}$-problem. Indeed, for any feasible pair $(u,q)$, replacing $q$
by
\[
\widetilde q=[u_1^\top y_1]_+
\]
preserves the state trajectory and all constraints, while satisfying
$\widetilde q\le q$ almost everywhere. Since the extraction constraint
implies $z_2(\Tf)\ge p\Tf>0$, this replacement cannot decrease the
objective $z_2(\Tf)/z_+(\Tf)$. Consequently, an optimal solution can be
chosen such that
\[
q=[u_1^\top y_1]_+
\quad\text{almost everywhere},
\]
and therefore $z_+(\Tf)=E_{1,+}$.

For every feasible finite-horizon trajectory, the available-storage
inequality and the extraction constraint yield
\[
z_1(\Tf)
=
E_1
\ge
E_2-S_a(x_0)
\ge
p\Tf-S_a(x_0)>0.
\]
Moreover,
\[
z_+(\Tf)=E_{1,+}\ge E_1.
\]
Thus, the denominators of both fractional objectives are uniformly
bounded away from zero, and the terminal costs
\[
(z_1,z_2)\mapsto\frac{z_2}{z_1},
\qquad
(z_+,z_2)\mapsto\frac{z_2}{z_+}
\]
are continuous on the corresponding feasible terminal sets. The
existence theorem in~\cite{LeeMarkus1967} therefore yields an optimal
control for OCP~\eqref{eq:OCP_Efficiency} for each
$\phi\in\{\eta,{\eta_{+}}\}$.

For the cyclic problem, periodicity and passivity imply
\[
E_1\ge E_2\ge p\Tp>0,
\qquad
E_{1,+}\ge E_1.
\]
Hence, both cyclic fractional objectives are continuous on their
feasible sets. In addition, the periodicity condition
\[
x(0)=x(\Tp)
\]
defines a closed endpoint constraint, while the initial state ranges
over the compact set $\mathbb X$. Applying the same existence argument
therefore yields an optimal cyclic control for
OCP~\eqref{eq:OCP_clc_Eff} for each
$\phi\in\{\eta,{\eta_{+}}\}$.
\End
\end{pf}

\bibliography{ALL_GP}
\end{document}